\documentclass[11pt]{article}
\usepackage{amssymb, amsmath}
\usepackage[margin=1in]{geometry}

\usepackage{lineno}
\usepackage{amsmath}
\usepackage{amsthm}
\usepackage{mathtools}
\usepackage{mathrsfs}
\usepackage{amssymb}
\usepackage{verbatim}
\usepackage{algorithmic}
\usepackage[linesnumbered,ruled,vlined]{algorithm2e}
\usepackage{color}
\usepackage[dvipsnames]{xcolor}
\usepackage{mdframed}
\usepackage{marginnote}
\usepackage{amsfonts}
\usepackage{lineno}

\newcommand*\patchAmsMathEnvironmentForLineno[1]{%
  \expandafter\let\csname old#1\expandafter\endcsname\csname #1\endcsname
  \expandafter\let\csname oldend#1\expandafter\endcsname\csname end#1\endcsname
  \renewenvironment{#1}%
     {\linenomath\csname old#1\endcsname}%
     {\csname oldend#1\endcsname\endlinenomath}}%
\newcommand*\patchBothAmsMathEnvironmentsForLineno[1]{%
  \patchAmsMathEnvironmentForLineno{#1}%
  \patchAmsMathEnvironmentForLineno{#1*}}%
\AtBeginDocument{%
\patchBothAmsMathEnvironmentsForLineno{equation}%
\patchBothAmsMathEnvironmentsForLineno{align}%
\patchBothAmsMathEnvironmentsForLineno{flalign}%
\patchBothAmsMathEnvironmentsForLineno{alignat}%
\patchBothAmsMathEnvironmentsForLineno{gather}%
\patchBothAmsMathEnvironmentsForLineno{multline}%
}
\usepackage[shortlabels]{enumitem}
\usepackage{url}
\usepackage[pdftex, colorlinks,
  linkcolor=blue,
  citecolor=blue,
  filecolor=blue,urlcolor=blue]
{hyperref}
\usepackage[capitalise]{cleveref}
\newtheorem{theorem}{Theorem}[section]
\newtheorem{lemma}[theorem]{Lemma}
\newtheorem{definition}[theorem]{Definition}
\newtheorem{proposition}[theorem]{Proposition}

\newtheorem{remark}[theorem]{Remark}

\newcommand{\GB}{\text{Gr\"{o}bner} }
\newcommand{\poly}{\mathrm{poly}}

\newcommand{\sos}{\text{\sc{SoS}}}
\newcommand{\x}{\mathbf{x}}

\newcommand{\MS}{\mathcal{S}}

\usepackage[normalem]{ulem}

\title{On the Bit Size of Sum-of-Squares Proofs for \\Symmetric Formulations \thanks{This work was funded by the Swiss National Science Foundation project No.200021\_207429 / 1 \textit{Ideal Membership Problems and the Bit Complexity of Sum of Squares Proofs}}} 
\date{}

\author{Alex Bortolotti\thanks{University of Applied Sciences and Arts of Southern Switzerland, IDSIA, Lugano, Switzerland. E-mail: \href{mailto:alex.bortolotti@supsi.ch}{\texttt{alex.bortolotti@supsi.ch}}, \href{mailto:marilena.palomba@supsi.ch}{\texttt{monaldo.mastrolilli@supsi.ch}},  \href{mailto:monaldo.mastrolilli@supsi.ch}{\texttt{marilena.palomba@supsi.ch}}, \href{mailto:luis.vargas@supsi.ch}{\texttt{luis.vargas@supsi.ch}}.} \  \and Monaldo Mastrolilli\footnotemark[2] \ \and Marilena Palomba\footnotemark[2]  \ \and Luis Felipe Vargas\footnotemark[2]
}

\begin{document}

\maketitle

\begin{abstract}
    The Sum-of-Squares ($\sos$) hierarchy is a powerful framework for polynomial optimization and proof complexity, offering tight semidefinite relaxations that capture many classical algorithms. 
    Despite its broad applicability, several works have revealed fundamental limitations to $\sos$ automatability.
    (i) While low-degree $\sos$ proofs are often desirable for tractability, recent works have revealed they may require coefficients of prohibitively large bit size, rendering them computationally infeasible.
    (ii) Prior works have shown that $\sos$ proofs for seemingly easy problems require high-degree. In particular, this phenomenon also arises in highly symmetric problems.
    Instances of symmetric problems--particularly those with a small number of constraints--have repeatedly served as benchmarks for establishing high-degree lower bounds in the $\sos$ hierarchy.
    It has remained unclear whether symmetry can also lead to large bit sizes in $\sos$ proofs, potentially making low-degree proofs computationally infeasible even in symmetric settings.
    
    In this work, we resolve this question by proving that symmetry alone does not lead to large bit size $\sos$ proofs. Focusing on symmetric Archimedean instances, we show that low-degree $\sos$ proofs for such systems admit compact, low bit size representations.
    Together, these results provide a conceptual separation between two sources of $\sos$ hardness--degree and bit size--by showing they do not necessarily align, even in highly symmetric instances. This insight guides future work on automatability and lower bounds: symmetry may necessitate high-degree proofs, but it does not by itself force large coefficients. 
\end{abstract}

\newpage

\section{Introduction}\label{sec:Introduction}

The Sum-of-Squares ($\sos$) hierarchy, also known as the Lasserre hierarchy \cite{Lasserre2001,Parrilo2000StructuredSP}, is one of the most powerful and broadly applicable frameworks for algorithm design and complexity analysis in polynomial optimization. It systematically generates increasingly tighter semidefinite programming (SDP) relaxations and subsumes many classical algorithms, see e.g. \cite{FlemingKothariPitassi19, Lasserre2001, Laurent2009,  Parrilo03}. Over the past two decades, $\sos$ has played a central role in advancing our understanding of both algorithmic upper and lower bounds and proof complexity. 
However, despite its generality, a growing body of work has uncovered inherent limitations of the hierarchy that has emerged in the last years. Indeed, it is only relatively recently that O'Donnell \cite{odonnell2017} and Raghavendra and Weitz \cite{raghavendra_weitz2017} have demonstrated that efficiently computable, i.e. in polynomial time, low-degree $\sos$ proofs might be impossible to obtain due to their inherently high-bit size.

It is by now well-understood that certain structural properties—such as high symmetry or compact constraint descriptions—can significantly influence the degree complexity of $\sos$. In particular, instances with a small number of symmetric constraints have often served as benchmarks in establishing $\sos$ high degree lower bounds \cite{Grigoriev01, GrigorievV01, KLM-unbounded-SOS-hierarchy-integrality-gap, KLM-hard-prob-formulation, KLM-tight-SOS-LB-binary-POP, KLM-SOS-hierarchy-LB-symmetric-formulations,kurpisz_et_al:SOS-certification-symmetric-quadratic-functions,laurent2003lower,potechin:SOS-LB-from-symmetry}. 
Examples of this kind include the (infeasible) \textsc{Knapsack} problem defined by $\{x \in \{0,1\}^n \ | \ \sum_{i=1}^n x_i = \frac{n}{2}\}$ with $n$ odd, for which Grigoriev showed that degree-$\Omega(\lfloor \frac{n}{2} \rfloor)$ $\sos$ proofs are necessary for certifying that the instance is unsatisfiable \cite{Grigoriev01}. In fact, symmetric problems have been shown to be among the most challenging in terms of the degree required by the 0/1 $\sos$ hierarchy \cite{KLM-hard-prob-formulation}: for instance, the symmetric problem \textsc{Min-Knapsack} exhibits an arbitrarily large integrality gap even at degree $n-1$.

In this work, we clarify and refine this understanding by addressing a fundamental and previously unresolved question: \emph{Does symmetry alone suffice to make $\sos$ hard due to high bit size?} More precisely, we focus on symmetric Archimedean instances defined by a polynomial number of constraints. Although certain special cases of symmetric Archimedean instances have previously appeared in high-degree lower bounds (see e.g. \cite{Grigoriev01}), the role of symmetry in the underlying source of $\sos$ hardness remains unclear, especially with respect to the bit size and the succinct representation of refutations \cite{Weitz:Phd}.

\textbf{Our contribution.} Our main contribution is to rigorously rule out a natural but previously open possibility: that symmetric instances could be hard for $\sos$ due to the bit size rather than degree. We show that this is not the case. Specifically, we demonstrate that low degree proofs of instances under symmetry conditions with a polynomial number of constraints have low bit sizes. This result provides a conceptual clarification of the role of symmetry in $\sos$ lower bounds: while symmetry can make $\sos$ fail at low degrees, it does not, in itself, force high-bit size solutions.

Our approach is based on representation simplification techniques that exploit structural properties of Archimedean systems and Gröbner bases. We use tools from convex geometry and polynomial ideal theory, together with concepts specific to $\sos$ proofs such as pseudoexpectations, to reduce the complexity of $\sos$ representations.
These algebraic simplifications are then combined with symmetry reductions: by leveraging finite group actions, we restrict the $\sos$ proof search to low-dimensional invariant subspaces. This results yield to $\sos$ proofs with ``small'' coefficients. 
Thus, any obstacle to solving these problems via $\sos$ cannot arise from high-bit sizes, but must be fundamentally combinatorial or algebraic in nature.

This insight has several implications. First, it separates two common sources of complexity in $\sos$--degree and bit size--by showing that they do not necessarily align, even in structured, highly symmetric cases. Second, it provides guidance for future lower-bound constructions: symmetry alone does not lead to hard-to-represent proofs, and so other mechanisms must be invoked when designing instances hard for $\sos$ in both degree and bit size. Finally, it reinforces the importance of degree as the primary complexity parameter in understanding the limitations of the $\sos$ hierarchy on symmetric instances.

\textbf{Related literature.}  Extensive research has explored the symmetric properties of $\sos$ in polynomial optimization \cite{BlekhermanR21, ChoiTR87, CimpriKS08, DebusR23, GATERMANN200495}. We refer to \cite{MoustrouRV23} for a thorough review on the topic. This literature primarily addresses algebraic aspects and implementation benefits, offering limited insight into computational complexity and no results concerning bit complexity analysis. A notable result is due to Riener et al. \cite{RienerTJL11}, who proved that the size of the matrix needed to find a low-degree sum-of-squares representation of an unconstrained homogeneous symmetric polynomial is independent from $n$. We emphasize that our setting is more general, allowing for the search, in symmetric frameworks, of $\sos$ proofs of nonhomogeneous polynomials subject to a nonempty set of polynomial constraints.

Early approaches to systematically study the degree automatability of the $\sos$ proof system leverage algebraic proof systems and their simulation by $\sos$. Raghavendra and Weitz~\cite{raghavendra_weitz2017} obtained a sufficient condition based on the \emph{Nullstellensatz} proof system, recently improved by Bortolotti et al. \cite{BortolottiMV25ICALP} who extended this to \emph{Polynomial Calculus} and introduced the first criterion for bounded-coefficient $\sos$ refutations.
Later progress has identified structured settings where $\sos$ relaxations remain tractable (here, $\sos$ relaxations refer to the semidefinite programming hierarchy that approximates polynomial optimization problems by searching for $\sos$ certificates of nonnegativity). Gribling et al. \cite{Gribling23} showed polynomial-time solvability under strong algebraic and geometric assumptions for systems with inequality constraints and full-dimensional feasibility. Palomba et al. \cite{palomba} independently showed that $\sos$ bounds for certain copositive programs can be computed efficiently.

\subsection{Technical overview}\label{sec:our-contribution}

\textbf{Preliminaries.} Let $\mathbb{R}[x_1,\ldots,x_n]$ denote the ring of $n$-variate real polynomials and let $\mathbb{R}[x_1,\ldots,x_n]_d$ be the vector space of polynomials of degree at most $d$. Further, we denote as $\Sigma$ the convex cone of polynomials that can be decomposed into a $\sos$ of polynomials, and we set $\Sigma_{2d} = \Sigma \cap \mathbb{R}[x_1, \ldots, x_n]_{2d}$. 

Let $\mathcal P = \{p_1=0, \dots p_m=0\}$ and $\mathcal{Q}=\{q_1\geq 0, \dots, q_{\ell}\geq 0\}$ be sets of polynomial equality and inequality constraints, respectively. We define the associated semialgebraic \emph{zero set} as $S = \{x \in \mathbb{R}^n \ | \ p_i(x)=0 \ \forall i \in [m] \text{ and } q_i(x) \geq 0 \ \forall j \in [\ell]\}$. Given a polynomial $r\in \mathbb R[x_1, \ldots, x_n]$, a \emph{sum-of-squares proof} of nonnegativity of $r$ over $S$ from $(\mathcal{P,Q})$ consists of an identity
\begin{equation}\label{eq:def_sos_proof}
    r = s_0 + \sum_{i=1}^m h_i p_i + \sum_{j=1}^{\ell} s_jq_j,
\end{equation}
where $s_0, s_1, \ldots, s_{\ell} \in \Sigma$ and $h_1, \ldots, h_m \in \mathbb{R}[x_1, \ldots, x_n]$. An $\sos$ proof of nonnegativity of the polynomial $r=-1$ from $(\mathcal{P},\mathcal{Q})$ is called an \emph{$\sos$ refutation} of $(\mathcal{P},\mathcal{Q})$; it certifies that the constraint set $\mathcal{P} \cup \mathcal{Q}$ is unsatisfiable. The \emph{degree} of the $\sos$ proof is the maximum degree of the polynomials appearing in \eqref{eq:def_sos_proof}, while the \emph{bit size} refers to the length of the binary representation of the proof under some standard encoding of rational coefficients. 
Furthermore, in what follows we assume the inputs $r, \mathcal{P}, \mathcal{Q}$ to have bit size polynomial in $n$. 

We are interested in understanding the automatability of $\sos$ proofs of a fixed degree $d\in O(1)$. The problem of finding a degree-$d$ $\sos$ proof can be formulated as a semidefinite program (SDP) of size $n^{O(d)}$, leveraging the well-known correspondence between $\sos$ polynomials and positive semidefinite (PSD) matrices (see, e.g., \cite{Laurent2009}). Based on this formulation, it has often been claimed that such feasibility SDPs can be solved in time $n^{O(d)}$ using the ellipsoid method.

However, in a recent work, O’Donnell \cite{odonnell2017} challenged this widely repeated claim. He constructed systems of polynomial inequalities with bounded coefficients for which all degree-2 $\sos$ certificates require doubly-exponential-sized coefficients. As a consequence, any $\sos$ proof must involve exponentially many bits, implying that the ellipsoid method will require exponential time to solve the corresponding SDP.

We aim to study whether a given triple $(r, \mathcal{P}, \mathcal{Q})$ satisfies the following property:
\begin{enumerate}
        \item[(P)] Assume there exists a degree-$d$ $\sos$ proof of $r$ from $(\mathcal{P}, \mathcal{Q})$ (as in (\ref{eq:def_sos_proof})). Then there exists another such proof of degree $d$ in which all coefficients are bounded by $2^{\poly(n^d)}$.
\end{enumerate}
As shown by O’Donnell \cite{odonnell2017} (see also \cite{Hakoniemi-PhD} for a more detailed exposition), property (P), together with the assumption that the constraint set $(\mathcal{P}, \mathcal{Q})$ is Archimedean, implies that $\sos$ proofs can be efficiently found. Specifically, if a degree-$d$ $\sos$ proof of $r$ exists, then for any rational $\varepsilon > 0$, one can efficiently compute a degree-$d$ $\sos$ proof of $r + \varepsilon$ from $(\mathcal{P}, \mathcal{Q})$ in time polynomial in $n$ and $\log(1/\varepsilon)$. We note that the additive error $\varepsilon$ arises from the numerical nature of semidefinite programming: the ellipsoid method can only determine the feasibility up to a small additive error. This is generally not considered problematic as the error can be tightly controlled.

\textbf{Our results.}
Although symmetry has been linked to high-degree lower bounds in $\sos$, we prove that it does not inherently cause large coefficients: symmetric systems admitting degree-$d$ proofs also admit representations with coefficients bounded by $2^{\poly(n^d)}$.
Specifically, in \cref{th:O(1)_symmetric_constraint_sos-proofs}, we show that if $G$ is a direct product of $O(1)$ symmetric groups, then for any $G$-invariant polynomial system $\mathcal{P} \cup \mathcal{F}$--with a polynomial number of equality constraints and $\mathcal{F}$ Gröbner basis--any degree-$2d$ $\sos$ proof admits a representation with coefficients bounded by $2^{\poly(n^d)}$.
In \cref{th:invariant_systems_refutation}, we establish a similar result for refutations: if $\mathcal{P}$ is a $G$-invariant system of polynomial equalities over a finite domain $\mathcal{D}$ and $\mathcal{P}\cup \mathcal{D}$ admits a degree-$2d$ $\sos$ refutation, then it also admits a degree-$2d$ refutation with coefficients bounded by $2^{\poly(n^d)}$.

With this aim, we first establish a structural result for Archimedean systems. We show that given an Archimedean pair $(\mathcal{P}, \mathcal{Q})$ and a set $\mathcal{R}$ of additional equality constraints, any degree-$2d$ $\sos$ refutation of $(\mathcal{P} \cup \mathcal{R}, \mathcal{Q})$ can be transformed into a degree-$O(d)$ refutation in \emph{normal form}, where each term $h_i r_i$, for $r_i \in \mathcal{R}$, can be assumed to take the form $\alpha_i r_i^2$ for scalars $\alpha_i \in \mathbb{R}$. This generalizes a result of Hakoniemi \cite{Hakoniemi-PhD}, originally proven for Boolean systems, to the broader Archimedean setting--i.e., systems where boundedness of the solution set can be $\sos$ certified. 
This normal form result will play a central role in the proof of our main results by enabling a precise control over the number of variables in the semidefinite programs characterizing $\sos$ proofs under symmetry. 
This is key to applying structural results such as \cref{th:porkolab-invariant-SDP+LP}, ultimately leading to polynomial bounds on the bit size of $\sos$ refutations, as established in \cref{th:invariant_systems_refutation}.

\textbf{Structure of the paper.}
In \cref{sec:sos-reductions}, we introduce reduction techniques for simplifying $\sos$ refutations over Archimedean systems. This section culminates in \cref{th:normal_form_archimedean_systems}, which establishes the normal form for $\sos$ refutations.  In \cref{sec:invariant-sos-and-finite-orbits}, we develop the symmetry framework by analyzing group actions on polynomials and bounding the number of resulting orbits. These bounds allow us to reduce the dimension of the semidefinite programs used to encode $\sos$ proofs. The main results are presented in \cref{sec:symmetric-case}, where we show that for systems with a polynomial number of constraints, under some symmetry assumptions, any low-degree $\sos$ proof or refutation can be rewritten with coefficients of polynomial bit size.

\section{Sums-of-Squares reductions}\label{sec:sos-reductions}

The focus of this section is on reduction techniques that exploit polynomial system structure for simplifying $\sos$ refutations.
We begin by introducing a normal form for $\sos$ refutations in the setting of Archimedean pairs. Recall that a pair $(\mathcal{P}, \mathcal{Q})$, where $\mathcal{P}$ is a set of polynomial equality constraints and $\mathcal{Q}$ is a set of polynomial inequality constraints, is Archimedean if there exists $N \in \mathbb{N}$ such that $N - \sum_{i=1}^n x_i^2$ has an $\sos$ proof from $(\mathcal{P},\mathcal{Q})$, which essentially implies that the associated semialgebraic set is ``provably'' bounded (see e.g. \cite{Laurent2009}).
We then focus on systems of polynomial equalities, demonstrating how reductions by a Gröbner basis provide a canonical representation for $\sos$ proofs modulo the ideal generated by the equalities. Crucially, we show how to convert a reduced proof back into a standard $\sos$ refutation.
These reduction techniques are essential tools for the analysis and proofs presented in the subsequent sections.

\subsection{$\sos$ refutations over Archimedean systems}

In \cite{Hakoniemi-PhD}, Hakoniemi shows an interesting structural property of $\sos$ refutations in the Boolean setting. For a system of polynomial equalities $\mathcal{P} = \{p_1 = 0, \ldots, p_m = 0\}$, alongside the Boolean constraints $x_i^2 - x_i =0$ for each variable $x_i$, any $\sos$ refutation initially expressed as $-1 = \sum s_i^2 + \sum h_i p_i + \sum r_i (x_i^2 - x_i)$, where $s_i,h_i,r_i$ are polynomials, can be shown to exhibit an alternative form, also called \emph{normal} form:
\begin{align*}
    -1 = \sum_{i=1}^t \tilde{s_i}^2 + \sum_{i=1}^m \alpha_i p_i^2 + \sum_{i=1}^n \tilde{r_i} (x_i^2 - x_i)
\end{align*}
where, notably, the coefficients $\alpha_i$ are scalars, i.e. $\alpha_i \in \mathbb{R}$.

This section extends Hakoniemi's work on $\sos$ refutations. We move beyond Boolean constraints to consider systems containing Archimedean pairs $(\mathcal{P},\mathcal{Q})$, a core concept in real algebraic geometry, in particular regarding Positivstellensatz results, and the moment-$\sos$ hierarchy (see also \cite{Laurent2009,Marshall2008}). Further, in \cref{sec:symmetric-case}, we will use normal forms to construct simpler $\sos$ refutations that allow us to bound their coefficients.

We begin by recalling some fundamental notions of convex sets in vector spaces, including the \emph{separation theorem for cones} (see e.g. \cite{BoydV04}).

\begin{definition}[Convex cones and order units]
    Let $V$ be an $\mathbb{R}$-vector space. A subset $C\subseteq V$ is called a \emph{convex cone} if $0\in C$, $C+C\subseteq C$ and $\mathbb{R}_+C\subseteq C$. We say that $C$ is proper if $C\neq V$.
    Furthermore, a point $u \in V$ is a \emph{order unit} for the convex cone $C$ (in $V$) if, for every $x\in V$, there exists some $N\in \mathbb{N}$ such that $Nu+x \in C$.
\end{definition}

\begin{theorem}[Isolation theorem for cones]\label{th:isola}
    Let $u$ be an order unit for the proper convex cone $C$ in the $\mathbb{R}$-vector space $V$. Then, there exists a linear functional $L: V\to \mathbb{R}$ such that $L(u)=1$ and $L(C)\subseteq \mathbb{R}_+$.    
\end{theorem}

Next, we introduce the real algebraic notions of semialgebraic sets and the cone of polynomials provably positive via $\sos$. Let $\mathcal{P}=\{p_1, \dots, p_m\}$ and $\mathcal{Q}=\{ q_1, \dots, q_{\ell} \}$ be two sets of $n$-variate polynomials. The \emph{semialgebraic set} generated by the pair $(\mathcal{P}, \mathcal{Q})$ is
\begin{equation*}
    K=\{x\in \mathbb{R}^n \ | \ p_i(x)=0 \text{ for } i\in [m] \ \text{ and } q_j(x)\geq0 \text{ for } j\in[\ell]\}.
\end{equation*}
Our objective is to study polynomials that are nonnegative on $K$. Let $k \in \mathbb{N}$ and set $q_0 := 1$, then the \emph{$2k$-truncated quadratic module} is defined as
\begin{equation*}
    \mathcal{M}(\mathcal{P,Q})_{2k}:=\left\{\sum_{i=1}^mh_ip_i + \sum_{j=0}^{\ell} s_j q_j \ | \ s_j \in \Sigma, h_i \in \mathbb{R} \ \text{ s.t. } \deg(s_j q_j), \deg(h_i p_i) \leq 2k \right\}.
\end{equation*}
It is the set of polynomials that admit a degree-$2k$ $\sos$ proof from $(\mathcal{P},\mathcal{Q})$.

\begin{definition}\label{def:2k-archimedean}
    We say the the pair of polynomials sets $(\mathcal{P}, \mathcal{Q})$ is degree-$2k$ Archimedean if there exists $N\in \mathbb{N}$ such that $N-\sum_{i=1}^nx_i^2 \in \mathcal{M}(\mathcal{P}, \mathcal{Q})_{2k}$.
\end{definition}

As an immediate consequence of this definition, we have the following useful lemma.

\begin{lemma}\label{th:archimedeanity_implies_order_of_unit}
    Assume $(\mathcal{P}, \mathcal{Q})$ is degree-$2k$ Archimedean for some $k \in \mathbb{N}$. Then, for every polynomial $p$ of degree $2d$ there exists $N \in \mathbb{N}$ such that $N- p\in \mathcal{M}(\mathcal{P}, \mathcal{Q})_{2(d+k-1)}$.
\end{lemma}
\begin{proof}
    It suffices to show that for any monomial $m$ of degree at most $2d$ there exists $N'$ such that $N'\pm m\in \mathcal{M}(\mathcal{P}, \mathcal{Q})_{2(d+k-1)}$. Let $N_k$ be as $N$ in Definition \ref{def:2k-archimedean}. We first show the following claim.
    
    \textbf{Claim (1):}  Let $m_1$ be a monomial of degree at most $d$. Then there exists $N'$ such that $N' \pm m_1^2\in \mathcal{M(P,Q})_{2(d+k-1)}$
    
    \textbf{Proof of Claim (1).}We proceed by induction on $d$. For $d=1$, we have $m=x_i$ for some $i\in[n]$, and thus $ N_k - x_i^2 =N_k - 
    \sum_{i=1}^nx_i^2 +\sum_{j\neq i}x_j^2\in \mathcal{M(P,Q)}_{2k}$. Clearly, we also have $N_k + x_i^2\in \mathcal{M(P,Q)}_{2k}$. Now we assume the claim holds for all monomials with degree at most $d$. Let $m_1$ with $\deg(m_1)=d+1$, so that $m_1^2=x_i^2m_2^2$, for some $i\in [n]$ and some monomial $m_2$ with $\deg(m_2) = d$. By the induction hypothesis, there exists $\tilde{N}$ such that $\tilde{N}-m_2^2 \in \mathcal{M(P,Q)}_{2(d+k-1)}$. We set $N=\max\{N_k, \tilde{N}\}$ and we have the following identity
    \begin{align*}
       N^2-x_i^2m_2^2 = (N-m_2^2)x_i^2 + N(N-x_i^2), 
    \end{align*}
    which, by the induction hypothesis and under the given assumptions, shows that $N^2-x_i^2m_2^2\in \mathcal{M(P,Q)}_{2(d+k)}$. Clearly, we have that $N^2+x_i^2m_2^2\in\mathcal{M(P,Q)}_{2(d+k)}$, which concludes the proof of the claim. $\triangleleft$
    
    To conclude the proof of the lemma we consider a monomial $m$ of degree at most $2d$ and decompose it as $m=m_1m_2$, where $m_1$ and $m_2$ are monomials of degree at most $d$. By Claim (1), there exist natural numbers $N_1, N_2$ such that $N_1-m_1^2\in \mathcal{M}(\mathcal{P}, \mathcal{Q})_{2(d+k-1)}$ and $N_2-m_2^2 \in \mathcal{M}(\mathcal{P}, \mathcal{Q})_{2(d+k-1)}$. Then, for $N=\max\{N_1,N_2\}$, the following identities hold:
    \small{
    \begin{align*}
    \frac{1}{2}\Big[(1-m_1)^2 + (1-m_2)^2 + (1+m_1+m_2)^2 + 2(N-m_1^2) + 2(N-m_2^2)\Big]&= 2N+\frac{3}{2} + m_1m_2 \\
    \frac{1}{2}\Big[(1-m_1)^2 + (1+m_2)^2 + (1+m_1-m_2)^2 + 2(N-m_1^2) + 2(N-m_2^2)\Big] &= 2N+\frac{3}{2}-m_1m_2.
    \end{align*}}
    \normalsize
    This shows that there exists a natural number $N'$ such that $N'\pm m\in   \mathcal{M}(\mathcal{P}, \mathcal{Q})_{2(d+k-1)}$.
\end{proof}

Next, we introduce pseudoexpectations, a technical concept often useful for analyzing $\sos$ in proof complexity (see e.g.~\cite{BarakS14}). Crucially, under the mild condition of explicit boundedness--an assumption slightly stronger than Archimedeanity--the existence of a pseudoexpectation is equivalent to the nonexistence of an $\sos$ refutation for any given set of constraints \cite{BarakS14}. While this duality plays an important role for understanding $\sos$ refutations, in what follows we will rely only on one direction of the equivalence. Specifically, in \cref{th:duality_pseudoexpectations_and_refutations} we argue that the existence of a pseudoexpectation implies the absence of $\sos$ refutations.

\begin{definition}\label{def: pseudoexp}
    Consider the pair $(\mathcal{P,Q})$. A degree-$2d$ \emph{pseudoexpectation} for $(\mathcal{P}, \mathcal{Q})$ is a linear functional $L:\mathbb{R}[x_1, \ldots, x_n]_{2d}\to \mathbb{R}$ such that 
    \begin{itemize}
        \item $L(1)=1$.
        \item $L(p)\geq 0$ for every $p\in\mathcal{M}(\mathcal{P}, \mathcal{Q})_{2d}$.
    \end{itemize}
\end{definition}

\begin{remark}\label{th:duality_pseudoexpectations_and_refutations}
    Suppose a degree-$2d$ pseudoexpectation $L$ exists for $(\mathcal{P}, \mathcal{Q})$. We show that there is no degree-$2d$ $\sos$ refutation for $(\mathcal{P,Q})$. For the sake of contradiction, assume that there exists such a refutation of the form $-1 =s_0 + \sum h_i p_i + \sum s_iq_i$, where the $s_i$'s are sums of squares and the $h_i$'s are polynomials in $\mathbb{R}[x_1,\ldots,x_n]$. Then, by applying $L$ to both sides of the equality, we obtain that $-1 = L(-1) = L\left( s_0 \right) + L\left( \sum h_i p_i \right) + L \left( \sum s_i q_i \right)$. However, by \cref{def: pseudoexp}, it follows that the RHS of the equality is greater or equal to zero, thus leading to a contradiction. Therefore, the existence of a degree-$2d$ pseudoexpectation $L$ implies the nonexistence of a degree-$2d$ $\sos$ refutation of $(\mathcal{P}, \mathcal{Q})$. 
\end{remark}

\begin{theorem}\label{th:normal_form_archimedean_systems}
    Let $d$ and $k$ be fixed natural numbers such that $d\geq k\geq 1$. Consider a set of polynomial equalities $\mathcal{R}$ and let $(\mathcal{P,Q})$ be a degree-$2k$ Archimedean pair. If there exists a degree-$2d$ refutation of $(\mathcal{P \cup R,Q})$, then there exists also a refutation of the form 
    \begin{equation*}
        -1 = \sigma + \sum_{r\in \mathcal{R}}a_r r^2 + \sum_{p \in \mathcal{P}} h_p p + \sum_{q \in \mathcal{Q}} s_q q,
    \end{equation*}
    where $\sigma, s_q \in \Sigma$, $h_p \in \mathbb{R}[x_1, \ldots, x_n]$, and $a_r\in \mathbb{R}$ is a scalar. Further, the degrees of $\sigma, h_p p$, and $s_q q$ are all at most $2(d+k-1)$.
\end{theorem}

\begin{proof}
    Let $C$ be the set of degree-$2d$ polynomials that admit a degree $2(d+k-1)$ $\sos$ proof of the form $\sigma+ \sum_{r\in \mathcal{R}}a_r r^2 + \sum_{p \in \mathcal{P}} h_p p + \sum_{q \in \mathcal{Q}} s_q q$, where $a_r\in \mathbb{R}$ for $r\in \mathcal{R}$, and $h_p$ and $s_q$ are polynomials, for $p\in \mathcal{P}$ and $q\in \mathcal{Q}$. Then, $C$ is a convex cone in the vector space $V=\mathbb{R}[x_1, \ldots, x_n]_{2d}$. Furthermore, it follows from \cref{th:archimedeanity_implies_order_of_unit} that $u=1$ is a order unit of $\mathcal{M}(\mathcal{P,Q})_{2(d+k-1)}$, and, therefore, of $C$ (in $V$) as well.
    
    We proceed by contradiction. Suppose that $-1 \notin C$. This further implies that $C$ is a proper convex cone. By the isolation theorem (Theorem \ref{th:isola}), there exists a linear functional $L:\mathbb{R}[x_1, \ldots, x_n]_{2d}\to \mathbb{R}$ such that $L(1) = 1$ and $L(C) \subseteq \mathbb{R}_+$. In particular, this implies:
    \begin{itemize}
        \item $L(p)\geq 0$ for $p\in \mathbb{R}[x_1, \ldots, x_n]_{2d}\cap \mathcal{M}(\mathcal{P,Q})_{2(d+k-1)},$
        \item $L(r^2)=0$ for all $r\in \mathcal{R}$.
    \end{itemize}
    We will show that $L$ is a degree-$2d$ pseudoexpectation for $(\mathcal{P \cup R,Q})$. This, together with \cref{th:duality_pseudoexpectations_and_refutations}, implies that there is no degree-$2d$ $\sos$ refutation for the system $(\mathcal{P\cup R, Q})$, reaching a contradiction. 
    For this, it remains to show that $L(rm)=0$, where $r \in \mathcal{R}$ and $m$ is a monomial such that $\deg(rm)\leq 2d$. 
    Assume that $\deg(r)=d_0$ and decompose $m$ as $m=m_1m_2$ with $\deg(m_1)\leq d-d_0$ and $\deg(m_2)\leq d$. 
    
    We first prove that $L(m_1^2r^2)=0$. Since $(\mathcal{P,Q})$ is a degree-$2k$ Archimedean pair, there exists $N\in \mathbb{N}$ such that $N-m_1^2\in \mathcal{M}(\mathcal{P,Q})_{2(d-d_0+k-1)}$, and thus $Nr^2-m_1^2 r^2 \in \mathcal{M}(\mathcal{P,Q})_{2(d+k-1)}$. Then, we have $0\leq L(Nr^2 - m_1^2r^2) = -L(m_1^2r^2) \leq 0$.  Hence, $L(m_1^2r^2)=0$. 
    
    Next, let $a>0$ be a positive real number. Then, we have  
    \begin{equation*}
        0 \leq  L((m_1r \pm a m_2)^2) = L(m_1^2r^2) \pm 2aL(m_1m_2r) + a^2L(m_2^2)  =  \pm 2aL(m_1m_2r) + a^2L(m_2^2).
    \end{equation*}
    Then, we have that $\pm L(m_1m_2r) + \frac{a}{2}L(m_2)\geq 0$ for all $a>0$. This implies that $L(m_1m_2r) = L(mr) = 0$ as desired.
\end{proof}

\begin{remark}\label{ex:finite-domain}[Finite domain sets]
    Consider $x_1,\dots,x_n$ variables and let $k$ be a fixed integer. Let the \emph{finite domain set} be defined as $\mathcal D =\{ D_i = (x_i - \rho_1)(x_i - \rho_2)\cdots(x_i - \rho_{2k})\}_{i=1}^n$, for $\rho_j \in \mathbb{R}$. Note that each constraint $D_i(x_i) = 0$ enforces $x_i$ to take values in $\{\rho_1,\dots,\rho_{2k}\}$ for all $i$. It can be observed that $(\mathcal{D}, \emptyset)$ is a $2k$-Archimedean pair (see \cite{BortolottiMV25ICALP}).
\end{remark}

\begin{remark}[Dimension reduction in $\sos$ refutations]\label{rem:normal_form_SDP_dimension_reduction}
    The normal form established in \cref{th:normal_form_archimedean_systems} leads to a practical dimension reduction in $\sos$ refutations. Given an infeasible polynomial system $\mathcal{R}$ with a degree-$2d$ $\sos$ refutation of the form $-1 = \sigma + \sum_{r \in \mathcal{R}}\lambda_r r$, the standard formulation involves an SDP with up to $|\mathcal{R}|\binom{n + 2d}{2d}$ variables, due to the polynomial multipliers $\lambda_r$. However, by adding a ball constraint $\sum x_i^2 \leq M$, we obtain a $2$-Archimedean system, allowing for a refutation of the form $-1 = \tilde{\sigma} + \sum a_r r^2 + s \left( M - \sum x_i^2 \right)$, where $a_r \in \mathbb{R}$, $\tilde{\sigma},s \in \Sigma$ and $s$ has degree at most $2d - 2$. This reduces the number of variables in the SDP to $\binom{n + 2d}{2d} + |\mathcal{R}| + \binom{n + 2d - 2}{2d - 2}$. This decrease in the dimensionality of the problem is not sufficient for meaningful gains in a computational complexity sense on its own. Nevertheless, the possibility remains that this reduction could lead to sensible improvements in actual computation time during implementation. This goes beyond the scope of the present paper, and we defer this analysis for future work.
\end{remark}

\subsection{Gröbner bases reductions}

In this section, we simplify $\sos$ proofs by using polynomial division. We begin by giving basic notation and results related to polynomial division and Gröbner bases (see also \cite{Cox}).

Consider $\mathbb R[x_1,\dots,x_n]$ ordered according to any graded order. For simplicity, we will consider the \emph{graded lexicographic order} (\texttt{grlex}). Consider the polynomials $r, f_1, \ldots, f_t \in \mathbb{R}[x_1,\dots,x_n]$ and let $I = \langle f_1, \ldots, f_t \rangle$ be the ideal generated by the set of polynomials $\mathcal F = \{f_1, \ldots, f_t\}$. We denote by $\overline{r}$ the remainder of the polynomial division of $r$ by $\mathcal F$ and we say that $\overline{r}$ is the \emph{reduced form of $r$ by $\mathcal F$}. Note that, under the \texttt{grlex} order, it follows that $\deg(\overline{r}) \leq \deg(r)$.

Further, let $I \subseteq \mathbb{R}[x_1, \ldots, x_n]$ be an ideal. If the property ``$\overline{r} = 0$ if and only if $r \in I$'' holds, we say that $\mathcal{F}$ is a \emph{Gröbner basis of $I$}. Moreover, it is known that the remainder of the polynomial reduction of a polynomial $r$ by a Gröbner basis $\mathcal{F}$ is uniquely determined. The uniqueness is in general not guaranteed for arbitrary polynomial systems. Notably, the Boolean axioms $ \mathcal{B}_n = \{x_1^2 - x_1, \ldots, x_n^2 - x_n\}$ constitute a Gröbner basis for the ideal $\langle \mathcal{B}_n \rangle$ they generate and whose zero set is given by the binary Boolean hypercube $\{0,1\}^n$.

Applying polynomial reduction by a Gröbner basis provides a way to simplify $\sos$ refutations by yielding a canonical reduced form for $\sos$ proofs modulo the generated ideal. We have the following result.

\begin{lemma}\label{th:reduction}
    Let $\mathcal{P} = \{p_1 = 0, \ldots, p_m = 0\}$ be a set of polynomial equality constraints, $\mathcal F=\{f_1,\dots,f_t\}$ be a Gröbner basis (in \texttt{grlex} order) for the ideal $\langle \mathcal{F} \rangle$ and let $r\in \mathbb R[x_1,\dots,x_n]_{2d}$.
    \begin{itemize}
        \item \emph{Reduction.} Let $r = \sigma + \sum_{i=1}^m h_i p_i + \sum_{i=1}^t q_i f_i$ be a degree-$2d$ $\sos$ proof of $r$ from $\mathcal{P \cup \mathcal{F}}$. Then the identity $\overline{r} = \overline{\sigma} + \sum_{i=1}^m \overline{h_i p_i}$ holds. Further, the RHS has degree at most $2d$ and size polynomial in the size of $r$ and $\mathcal{P \cup F}$.
        \item \emph{Reconstruction.} Suppose there exists an $\sos$ $\sigma$ and polynomials $h_i$ satisfying $\overline{r} = \overline{\sigma} + \sum_{i=1}^m \overline{h_i p_i}$ with $\max \{\deg{\sigma},\deg{h_i p_i}\} \leq 2d$ and with total size $\ell$. Then there exists a degree-$2d$ $\sos$ proof $r = \sigma + \sum_{i=1}^m h_i p_i + \sum_{i=1}^t q_i f_i$ of size $\poly(\ell)$ with degree at most $2d$.
    \end{itemize}
\end{lemma}

\begin{proof}
    \emph{Reduction.} Since polynomial reduction by a Gröbner basis is uniquely defined and linear, it forms a well-defined linear function.
    Thus, the identity $\overline{r} = \overline{\sigma} + \sum_{i=1}^m \overline{h_i p_i}$ holds. Furthermore, since $\mathcal{F}$ is a Gröbner basis with respect to the \texttt{grlex} order, then the RHS has degree at most $2d$ and size polynomial in the inputs $r, \mathcal{P}\cup\mathcal{F}$ (see also \cite{BortolottiMV25ICALP, Cox}).
        
    \emph{Reconstruction.} Consider the identity $\overline{r} = \overline{\sigma} + \sum_{i=1}^m \overline{h_i p_i}$. We show that this identity can be ``reconstructed'' to an $\sos$ proof of $r$ while preserving the degrees of $\sigma$ and $h_i p_i$, with at most a polynomial increase in size. 

    First, observe that $\overline{\sigma} + \sum_{i=1}^m \overline{h_i p_i}$ can be seen as the (unique) remainder of polynomial reduction by $\mathcal{F}$. That is, there exist polynomials $q_1,\ldots, q_t$ with $\deg(q_i f_i)\leq 2d$ for $i\in [t]$, such that
    \begin{equation*}
        \sigma + \sum_{i=1}^m h_i p_i = \sum_{i=1}^t q_i f_i + \left( \overline{\sigma} + \sum_{i=1}^m \overline{h_i p_i} \right).
    \end{equation*}
    Rearranging, we obtain
    \begin{equation*}
        \sigma + \sum_{i=1}^m h_i p_i - \sum_{i=1}^t q_i f_i = \overline{\sigma} + \sum_{i=1}^m \overline{h_i p_i} = \overline{r}.
    \end{equation*}
    Similarly, there exist polynomial $\rho_1, \ldots, \rho_t$, with $\deg(\rho_i f_i) \leq 2d$ with $i \in [t]$ such that
    \begin{align*}
        r = \sum_{i=1}^t \rho_i f_i + \overline{r}.
    \end{align*}
    Thus, we obtain
    \begin{align*}
        r = \sigma + \sum_{i=1}^m h_i p_i + \sum_{i=1}^t(\rho_i - q_i)f_i.
    \end{align*}
    The degree and size bounds follow immediately from the polynomial division algorithm with respect to the \texttt{grlex} order (see also \cite{BortolottiMV25ICALP, Cox}).
\end{proof}

\section{Invariant $\sos$ and finite orbits}\label{sec:invariant-sos-and-finite-orbits}

In this section, we introduce the natural action of a permutation group $G$ on polynomials and state the main properties of such actions. Furthermore, we analyze the action of a direct product of symmetric groups on pairs of exponent vectors with bounded total degree, and we prove that for any fixed degree $d$, the number of resulting orbits is bounded and, in fact, remains constant when $n \geq 2d$.  These structural properties will play a key role in \cref{sec:symmetric-case}, where we use them to upper bound the number of variables required to construct $\sos$ proofs under symmetry assumptions, thereby allowing for a reduction in the dimensionality of the corresponding semidefinite program.

\subsection{Preliminaries on (finite) group actions}

\begin{definition}[Group action]\label{def:group-action}
    Let $G$ be a group and $X$ be a set. A \emph{group action} of $G$ on $X$ is a function $\alpha: G \times X \rightarrow X$, denoted by $\alpha(g,x) = g \cdot x$, that satisfies the following properties: $e \cdot x = x$ for all $x\in X$, where $e\in G$ is the identity element of $G$, and $g \cdot(h \cdot x)=(gh) \cdot x$ for all $g,h\in G$ and $x\in X$.
\end{definition}

\begin{definition}[Orbits, Stabilizers, and Fixed Points]
    Let a group $G$ act on a set $X$.
    \begin{itemize}
        \item For $x \in X$, the set $Orb_x = \{g \cdot x \in X \ | \ g \in G \}$ is called the \emph{$G$-orbit} of $x$.
        \item For $x \in X$, the set $Stab_x = \{ g \in G \ | \ g \cdot x = x \}$ is called the \emph{stabilizer} of $x$ in $G$.
        \item The set of all distinct $G$-orbits of $X$ is called the \emph{quotient set} of $X$ by $G$, denoted $X/G$.
        \item An element $x \in X$ is a \emph{fixed point} if its $G$-orbit consists only of $x$, i.e., $Orb_x = \{x\}$. Equivalently, $x$ is a fixed point if its stabilizer is the entire group $G$, i.e., $Stab_x = G$.
    \end{itemize}
\end{definition}

\begin{remark}
    The group action induces an equivalence relation $\sim$ on $X$, where $x_i \sim x_j$ if and only if $x_j = g \cdot x_i$ for some $g \in G$. The equivalence classes of this relation are precisely the $G$-orbits. In addition, for each $x \in X$, the stabilizer $Stab_x$ is a subgroup of $G$. Furthermore, for a finite group $G$ acting on a set $X$, the Orbit-Stabilizer Theorem relates the cardinalities of $G, Orb_x$ and $Stab_x$ as follows  
    \begin{equation*}
        |G| = |Orb_x| \cdot |Stab_x| \qquad \forall x\in X.  
    \end{equation*}
\end{remark}

In what follows, we consider a specific type of group action relevant to polynomials, i.e., the action of a permutation group on the variables of a polynomial, that we now define formally. Given multi-index $\alpha\in \mathbb N^n$, we let $|\alpha|:= \| \alpha\|_1=\sum_{i=1}^n\alpha_i$. For $x = (x_1, \ldots, x_n) \in \mathbb R^n$ and a multi-index $\alpha = (\alpha_1, \ldots, \alpha_n)$, we write $x^{\alpha} = x_1^{\alpha_1} \cdots x_n ^{\alpha_n}$.

\begin{definition}[Action on Monomials and Polynomials]\label{def:action-on-poly}
    Let $G$ be a finite group that can act on the indices $1, \ldots, n$ (e.g., a subgroup of $S_n$). For $g \in G$ and a multi-index $\alpha = (\alpha_1, \ldots, \alpha_n) \in \mathbb{N}^n$, we define the action of $g$ on $\alpha$ as $g \cdot \alpha = (\alpha_{g^{-1}(1)}, \ldots, \alpha_{g^{-1}(n)})$.
    For a monomial $m = x^{\alpha}$, the \emph{action of $G$ on $m$} is defined for any $g \in G$ as:
    \begin{align*}
        g \cdot m := x^{g \cdot \alpha} = x_1^{\alpha_{g^{-1}(1)}} \cdots x_n^{\alpha_{g^{-1}(n)}} = x_{g(1)}^{\alpha_1} \cdots x_{g(n)}^{\alpha_n}.
    \end{align*}
    This action extends linearly to the set of polynomials $\mathbb{R}[x_1, \ldots, x_n]$. Thus, for a polynomial $p(x) = \sum_{i=0}^m a_i x^{\alpha_i}$, the action of $g \in G$ on $p$ is defined as $g \cdot p = \sum_{i=0}^m a_i (g \cdot x^{\alpha_i})$.
    
    A polynomial p is said to be \emph{$G$-invariant} if $g \cdot p = p$ for all $g \in G$. If $G=S_n$, where $S_n$ is the symmetric group on $n$ elements, the polynomial is called \emph{symmetric}.
\end{definition}

Note that the group action on polynomials, as defined above, is compatible with the algebraic structure of the polynomial ring, i.e., it distributes over addition and respects multiplication. In other words, applying a group element to a sum or product of polynomials yields the sum or product of the transformed polynomials.

\begin{lemma}\label{th:group_action_additivity_multiplicativity}
    Let $G$ be a finite group acting on polynomials as defined above. For any two polynomials $p,q \in \mathbb{R}[x_1, \ldots, x_n]$ and any $g \in G$, the action satisfies:
    \begin{itemize}
        \item \textit{Additivity:} $g \cdot (p+q) = g \cdot p + g \cdot q$.
        \item \textit{Multiplicativity:} $g \cdot (pq)  = (g \cdot p)(g \cdot q)$.
    \end{itemize}
\end{lemma}
\begin{proof}
    Additivity follows directly from the linear extension of the action from monomials to polynomials.
    
    For multiplicativity, consider two monomials $x^{\alpha}$ and $x^{\beta}$. Their product is $x^{\alpha + \beta}$. The action of $g$ on the product is $g \cdot x^{\alpha + \beta}=x^{g \cdot (\alpha + \beta)}$.
    Observe that $g \cdot (\alpha + \beta) = (\alpha_{g^{-1}(1)} + \beta_{g^{-1}(1)}, \ldots, \alpha_{g^{-1}(n)} + \beta_{g^{-1}(n)}) = g \cdot \alpha + g \cdot \beta$ component-wise. Therefore, for the product we have $(g \cdot x^{\alpha}) (g \cdot x^{\beta}) = x^{g \cdot \alpha} x^{g \cdot \beta} = x^{g \cdot \alpha + g \cdot \beta} = x^{g \cdot (\alpha + \beta)} = g \cdot x^{\alpha + \beta} = g \cdot (x^{\alpha} x^{\beta})$. Multiplicativity for general polynomials follows by linearity. 
\end{proof}

\begin{remark}
    Although we define group actions on polynomials via subgroups of the symmetric group $S_n$, acting by permuting variables, this framework extends naturally to arbitrary finite groups. By Cayley’s Theorem, every finite group $G$ is isomorphic to a subgroup of $S_{|G|}$, implying that any finite group action can be represented as a permutation action. This justifies our focus on permutation subgroups of $S_n$ when considering actions on polynomials in $n$ variables.
\end{remark}

Group actions provide a formal framework for describing symmetries and are particularly useful in the study of polynomials. We introduce some algebraic notions and state the properties that will be used in the proofs of our main results (see also \cite{DummitF04}).

Let now $\MS^n$ be the space of real symmetric matrices of dimension $n\times n$, endowed with the inner product $\langle X,Y\rangle:=\mathrm{Tr}(XY)$. A matrix $Q\in \MS^n$ is \emph{positive semidefinite}, denoted $Q \succeq 0$, if $x^TQx\geq 0$ for all $x\in \mathbb{R}^n$.
Consider a polynomial $p \in \mathbb R[x_1,\dots,x_n]_{2d}$ for fixed $d \in \mathbb{N}$. As observed in \cite{CTR1994}, there exists $Q \in \MS^{\omega_n^d}$ such that $p$ can be written as $p = \langle Q,\x_d\x_d^\top\rangle$, where $\omega_n^d= \binom{n+d}{d}$ is the number of elements in the monomial basis of degree at most $d$. In some cases, the symmetric matrix $Q$ has some useful properties, as shown below.

\begin{definition}\label{def:action-on-matrices}
    Let $S_n$ be the symmetric group of permutations of $n$ elements. Every permutation $\pi \in S_n$ of the indices of $x_1, \dots, x_n$ induces a permutation $\pi'\in S_{\omega_n^d}$ on the monomials in the monomial basis $\x_d$. Consider the permutation matrix $P_{\pi'}$ associated to $\pi'$, and let $Q\in \MS^{\omega_n^d}$ be a symmetric matrix whose entries are indexed by the monomial basis $\x_d$, i.e. $Q=(Q_{x^\alpha, x^\beta})$ for $x^\alpha, x^\beta$ entries of $\x_d$. The \emph{action} of $\pi \in S_n$ on $Q$ is given by
    \begin{equation*}
        \pi \star Q = P_{\pi'}Q P_{\pi'}^\top.
    \end{equation*}
    Equivalently, the entries of $ \pi \star Q $ satisfy $(\pi \star Q)_{x^\alpha, x^\beta} = Q_{\pi \cdot x^\alpha,\, \pi \cdot x^\beta}$. Moreover, for any $G$ subgroup of $S_n$, we say that $Q$ is \emph{$G$-invariant} if $\pi \star Q = Q$ for every $\pi \in G$.
\end{definition}

This definition gives rise to several important properties:

\begin{lemma}\label{lemma:action-on-PSD-matrix}
    Let $d\in O(1)$ be an integer and let $S_n$ be the symmetric group of $n$ elements.
    \begin{enumerate}
        \item For $Q$ positive semidefinite matrix and for every $\pi\in S_n$, $Q'=\pi \star Q$ is positive semidefinite. 
        \item For $p\in \mathbb R[x_1,\dots,x_n]_{2d}$ polynomial of degree at most $2d$, and $Q\in \MS^{\omega_n^d}$ such that $p=\langle Q,\x_d\x_d^\top \rangle$, the action of $\pi\in S_n$ on $p$ is given by $\pi \cdot p = \langle \pi \star Q, \x_d \x_d^\top\rangle$.
    \end{enumerate}
\end{lemma}
\begin{proof}
    (1) Let $Q \succeq 0$ be a positive semidefinite matrix. Let $\pi \in S_n$ be a permutation and consider the matrix $\pi \star Q = P_{\pi'} Q P_{\pi'}^\top$. 
    Since $P_{\pi'}$ is a permutation matrix, it is orthogonal, meaning that $P_{\pi'}^\top = P_{\pi'}^{-1}$. Let $v \in \mathbb{R}^{\omega_n^d}$ and define $w = P_{\pi'}^\top v$. Then we get
    \begin{equation*}
        v^\top (\pi \star Q) v = v^\top P_{\pi'} Q P_{\pi'}^\top v = (P_{\pi'}^\top v)^\top Q (P_{\pi'}^\top v) = w^\top Q w \geq 0,
    \end{equation*}
    since $Q \succeq 0$. We can also conclude $\pi \star Q \succeq 0$.

    (2) Suppose $p = \langle Q, \x_d\x_d^\top\rangle$ is a polynomial of degree at most $2d$. The action of $\pi \in S_n$ on the polynomial $p$ is defined by permuting the variables in the monomial basis $\x_d$, so $\pi \cdot \x_d = P_{\pi'} \x_d$. Then:
    \begin{align*}
        \pi \cdot p &= \langle Q, (\pi \cdot \x_d)(\pi \cdot \x_d)^{\top}\rangle = \langle Q, (P_{\pi'} \x_d) (P_{\pi'} \x_d)^{\top} \rangle \\
        &= \langle Q, P_{\pi'} \x_d \x_d^{\top} P_{\pi'}^{\top} \rangle = \langle P_{\pi'}^{\top} Q P_{\pi'}, \x_d \x_d^{\top} \rangle = \langle \pi \star Q, \x_d \x_d^{\top} \rangle,
    \end{align*}
    where the first equality follows from the properties of additivity and multiplicativity of group actions on polynomials, and the fourth equality follows from the characterization of Frobenius inner products as $\langle A, B \rangle = Tr(A^{\top}B)$ and the property of cyclicity of the trace.
\end{proof}

\subsection{Group actions on $\sos$ proofs}

In this section we examine how group actions interact with Sum-of-Squares proofs, establishing that invariance properties are preserved throughout the proof system. We demonstrate that if a polynomial is $G$-invariant, its $\sos$ representation can be chosen to respect this symmetry through an averaging construction using the Reynolds operator. This structural preservation enables us to work within the reduced-dimensional space of invariant polynomials, a key insight that will be crucial for bounding the bit complexity of symmetric $\sos$ proofs in \cref{sec:symmetric-case}.

\begin{proposition}\label{prop:invariant-sos}
    Let $p\in \mathbb{R}[x_1, \dots, x_n]_{2d}$ be a polynomial and let $G$ be a finite group. Assume that $p$ is $G$-invariant. Then, $p=\langle \overline{Q}, \x_d \x_d^\top\rangle$ for some $G$-invariant matrix $\overline{Q}$. 
    In addition, if $p$ is a sum of squares, then  $\overline{Q}$ can be taken positive semidefinite.
\end{proposition}

\begin{proof}
   Let $p\in \mathbb{R}[x_1, \dots, x_n]_{2d}$ be a polynomial and let $Q$ be such that $p = \langle Q, \x_d \x_d^\top \rangle$. Let now $G$ be a finite group acting on the variables $x_1, \dots, x_n$, and suppose that $p$ is invariant under the action of $G$.
    
    Let us define a new matrix $\overline{Q}$ as 
    \begin{equation}\label{eq:invariant-poly-avg}
        \overline{Q} := \frac{1}{|G|} \sum_{g \in G} g \star Q,
    \end{equation}
    where $g \star Q = P_{g'} Q P_{g'}^\top$ is the action of $g$ on $Q$ as in \cref{def:action-on-matrices}.
    Observe that matrix $\overline{Q}$ is invariant under the action of $G$. Indeed, for any $h \in G$, it holds
    \begin{equation*}
        h \star \overline{Q} = \frac{1}{|G|} \sum_{g \in G} h \star (g \star Q) = \frac{1}{|G|} \sum_{g \in G} (hg) \star Q = \overline{Q}, 
    \end{equation*}
    since the set $\{hg \mid g \in G\}$ is just a reindexing of $G$. 
    It also holds $p = \langle \overline{Q}, \x_d \x_d^\top \rangle$. In fact, since $p = \langle Q, \x_d \x_d^\top \rangle$ and $p$ is $G$-invariant, for all $g \in G$ we have 
    \begin{equation*}
        p = g \cdot p = \langle g \star Q, \x_d \x_d^\top \rangle. 
    \end{equation*}
    Therefore,
    \begin{equation*}
        p = \frac{1}{|G|} \sum_{g \in G} \langle g \star Q, \x_d \x_d^\top \rangle = \left\langle \frac{1}{|G|} \sum_{g \in G} g \star Q, \x_d \x_d^\top \right\rangle = \langle \overline{Q}, \x_d \x_d^\top \rangle. 
    \end{equation*}
    Finally, let now $\sigma\in \mathbb{R}[x_1, \dots, x_n]_{2d}$ be a $G$-invariant sum-of-squares and let $Q\succeq 0$ be such that $\sigma = \langle Q, \x_d \x_d^\top \rangle$. 
    Observe that each $g \star Q$ in \cref{eq:invariant-poly-avg} is then positive semidefinite since $Q \succeq 0$ and the conjugation by the orthogonal matrix $P_{g'}$ preserves positive semidefiniteness, as shown in \cref{lemma:action-on-PSD-matrix}. Since the sum of PSD matrices and scalar multiples of PSD matrices are still PSD, it follows that $\overline{Q} \succeq 0$.
\end{proof}

A convenient way to think of the construction in the proof of Proposition~\ref{prop:invariant-sos} is as follows.  Consider any polynomial $f\in\mathbb{R}[x_1,\dots,x_n]$, a map that sends $f$ to $\frac{1}{|G|}\sum_{g\in G} g\cdot f$ gives a linear projection of $\mathbb{R}[x_1,\dots,x_n]$ onto the subspace of $G$‑invariant polynomials. At the matrix level, this is exactly the map
\begin{equation*}
    Q \mapsto \overline Q = \frac{1}{|G|}\sum_{g\in G} g\star Q
\end{equation*}
which we applied above to construct a matrix invariant under the action of $G$.  This averaging map is known in invariant theory as the \emph{Reynolds operator}. We now give its formal definition.

\begin{definition}
    Let $G$ be a finite group acting on the polynomial ring $\mathbb R[x_1,\dots,x_n]$. The \emph{Reynolds operator} $R_G \colon \mathbb R[x_1,\dots,x_n] \to \mathbb R[x_1,\dots,x_n]$ is defined as
    \begin{equation*}
        R_G(f) = \frac{1}{|G|} \sum_{g\in G} g\cdot f.
    \end{equation*}
\end{definition}

\begin{remark}\label{rem:avg-sos-is-sos}
    Note that, given a polynomial $f\in \mathbb R[x_1,\dots,x_n]$, the polynomial $R_G(f)$ is $G$-invariant. Moreover, if $f$ is also $G$-invariant, then $f = R_G(f)$. In addition, the Reynolds operator preserves the sum-of-squares property, that is, if $f$ is a sum-of-squares polynomial, then $R_G(f)$ remains a sum of squares. More precisely, for $f$ sum of squares, there exists a positive semidefinite matrix $Q$ such that $f=\langle Q,\x_d \x_d^\top\rangle$. Then
    \begin{align*}
        R_G(f) = \frac{1}{|G|} \sum_{g\in G} g \cdot f = \frac{1}{|G|} \sum_{g\in G} \langle g \star Q,\x_d \x_d^\top\rangle = \left\langle \frac{1}{|G|} \sum_{g \in G} g \star Q, \x_d \x_d^\top \right\rangle = \langle \overline{Q}, \x_d \x_d^\top \rangle, 
    \end{align*}
    where $\overline{Q} = \frac{1}{|G|} \sum_{g \in G} g \star Q$. Since each matrix $g \star Q$ is positive semidefinite by \cref{lemma:action-on-PSD-matrix}, and the convex combination of positive semidefinite matrices remains positive semidefinite, it follows that $\overline{Q} \succeq 0$. Therefore, $R_G(f) = \langle \overline{Q}, \x_d \x_d^\top \rangle$ is itself a sum of squares.
\end{remark}

We now further expand our setting from $G$-invariant polynomials to $G$-invariant systems of polynomials that exhibit the same invariance property.

\begin{definition}[Invariant systems]\label{def:invsys}
    Let $G$ be a finite group, and let $\mathcal{P} \subseteq \mathbb{R}[x_1, \ldots, x_n]$ be a system of polynomials. We say that $\mathcal{P}$ is \emph{$G$-invariant} if it is closed under the action of $G$, that is, for every $g \in G$ and every $p \in \mathcal{P}$, we have $g \cdot p \in \mathcal{P}$. The set of $G$-orbits of $\mathcal{P}$ is denoted as $\mathcal{P}/G$.
\end{definition}

\begin{proposition}\label{th:invariant-system-averaged-sos}
    Let $G$ be a finite group and let $d$ be a fixed integer. Assume that the polynomials  $f,p_1, \dots, p_m$ and the set $\mathcal R=\{r_1, \dots, r_{\ell}\}$ are $G$-invariant. If there exists an $\sos$ proof of degree $2d$, $f=\sigma + \sum_{i=1}^m h_ip_i + \sum_{j=1}^{\ell} q_jr_j$, then there exists an $\sos$ proof of degree $2d$ of the form
    \begin{equation*}
        f=\tilde{\sigma} + \sum_{i=1}^m\tilde{h_i}p_i + \sum_{j=1}^{\ell} q_j'r_j,
    \end{equation*}    
    where the sum-of-squares $\tilde{\sigma}$ and each $\tilde{h}_i\in \mathbb R[x_1,\dots,x_n]$ are $G$-invariant. 
\end{proposition}

\begin{proof}
    Assume there exists an $\sos$ proof of polynomial $f$ of the form 
    \begin{equation}\label{eq:invariant-system-averaged-sos}
        f=\sigma + \sum_{i=1}^m h_ip_i + \sum_{j=1}^{\ell} q_jr_j.
    \end{equation} 
    Applying the Reynolds operator to both sides of identity (\ref{eq:invariant-system-averaged-sos}), we obtain the following identity
    \begin{equation*}
        R_G(f)= R_G(\sigma + \sum_{i=1}^m h_ip_i + \sum_{j=1}^{\ell} q_jr_j)
    \end{equation*}
    We recall that $R_G$ fixes every $G$-invariant polynomial (see \cref{rem:avg-sos-is-sos}). Further, since $f,p_1,\dots,p_m$ are $G$-invariant and $R_G$ is linear, we get
    \begin{align*}
        f = R_G(f) &= R_G(\sigma) + R_G(\sum_{i=1}^m h_ip_i ) + R_G(\sum_{j=1}^{\ell} q_jr_j)\\
        &= R_G(\sigma) + \sum_{i=1}^m \frac{1}{|G|}\sum_{g\in G}g\cdot( h_ip_i) + \sum_{j=1}^{\ell} R_G(q_jr_j)\\
        &= R_G(\sigma) + \sum_{i=1}^m \frac{1}{|G|}\sum_{g\in G}(g\cdot p_i)(g\cdot h_i) + \sum_{j=1}^{\ell} R_G(q_jr_j)\\
        &= R_G(\sigma) + \sum_{i=1}^m \frac{1}{|G|}p_i\sum_{g\in G}(g\cdot h_i) + \sum_{j=1}^{\ell} R_G(q_jr_j)\\
        &= R_G(\sigma) + \sum_{i=1}^m R_G(h_i) p_i + \sum_{j=1}^{\ell} R_G(q_jr_j)
    \end{align*}
    Then, the previous equation reduces to
    \begin{equation*}
        f=\tilde{\sigma} + \sum_{i=1}^m \tilde{h_i}p_i + \sum_{j=1}^{\ell} R_G(q_jr_j).
    \end{equation*}
    where $\tilde{\sigma} := R_G(\sigma)$, $\tilde{h}_i := R_G(h_i)$ for $i\in [m]$ are $G$-invariant and, by \cref{rem:avg-sos-is-sos}, $\tilde{\sigma}$ is a sum of squares.
    Consider now the sum
    \begin{equation*}
        \sum_{j=1}^{\ell} R_G(q_j r_j) = \sum_{j=1}^{\ell} \frac{1}{|G|} \sum_{g \in G} g \cdot (q_jr_j) = \sum_{j=1}^{\ell} \frac{1}{|G|} \sum_{g \in G} (g \cdot q_j)(g \cdot r_j),
    \end{equation*}
    where each $g \cdot r_j \in \mathcal R$ since $\mathcal R$ is $G$-invariant by assumption. 
    Thus $\sum_{j=1}^{\ell} R_G(q_j r_j)$ is a sum of the form $\sum_{k=1}^{\ell} q_k'r_k$ where $r_k\in \mathcal R$ and the polynomial coefficients $q_k'$ are given by
    \begin{equation*}
        q_k'= \frac{1}{|G|}  \sum_{j\in [\ell], g\in G, g \cdot r_j=r_k} (g \cdot q_j).
    \end{equation*}
\end{proof}

\subsection{Orbit counting under symmetry groups}\label{sect:counting_orbits}

The two technical lemmas presented in this section are key components in the proof strategy of our main results in \cref{sec:symmetric-case}. These lemmas are used to rigorously bound the number of variables involved in the semidefinite programs that characterize $\sos$ proofs and refutations under symmetry assumptions. Specifically, they enable us to leverage group symmetries to reduce the dimension of the SDP by counting the number of orbits under the group action. This orbit-counting argument is essential for ensuring that the dimension of the space in which the SDP solution lies is independent of the number of variables $n$. This invariance is precisely what allows us to apply \cref{th:porkolab-invariant-SDP+LP} and derive polynomial bounds on the bit size of $\sos$ proofs and refutations as formalized in \cref{th:O(1)_symmetric_constraint_sos-proofs} and \cref{th:invariant_systems_refutation}.

\begin{lemma}\label{th:orbits_pairs_indices_under_symmetric_action}
    Let $k=O(1)$ be a fixed positive integer. Let $n_1, n_2, \ldots, n_k \in \mathbb{N}$ be such that $\sum_{i=1}^k n_i = n$. Consider the group $G = S_{n_1} \times S_{n_2} \times \ldots \times S_{n_k}$, which acts on the indices $1,2,\ldots,n$ by permuting the indices within each block $1, \ldots, n_1, n_1+1, \ldots, n_1 + n_2, \ldots, n$. Consider the sets $W = \{x \in \mathbb{N}^n\ | \ \sum_{i=1}^n x_i \leq d \}$ and $Y = \{(x,y) \in \mathbb{N}^n \times \mathbb{N}^n \ | \ \sum_{i=1}^n x_i \leq d, \sum_{i=1}^n y_i \leq d \}$.
    Let $x\in W$ and  $(x,y) \in Y$, let $g \in G$ act on $x$ as in \cref{def:action-on-poly} and define the action of $g$ on $(x,y)$ as
    \begin{equation*}
        g \cdot (x,y) = (g \cdot x, g \cdot y) = \left( (x_{g(1)}, x_{g(2)}, \ldots, x_{g(n)}), (y_{g(1)}, y_{g(2)}, \ldots, y_{g(n)} \right).
    \end{equation*}
    Then, for a fixed nonnegative integer $d$, the number of orbits of the action of $G$ on $W$ and $Y$ is bounded by a constant that depends only on $d$. Specifically, for $n \geq 2d$, the number of orbits is constant with respect to $n$.
\end{lemma}

\begin{proof}
    We will prove the statement for the set $Y$ as the argument for $W$ follows the same ideas and strategy and is strictly simpler.
    First, we consider the case of the full symmetric group, i.e., assume $G = S_n$. Let $(x,y) = \left( (x_1, x_2, \ldots, x_n),(y_1, y_2, \ldots, y_n) \right)$ in $Y$ and consider its \emph{multiset of pairs}, that is, $\{(x_1,y_1), (x_2,y_2), \ldots, (x_n, y_n) \}$. 
    Further, let $g \in G$ and observe that the action of $g$ on $(x,y)$ only permutes the pairs of components $(x_i,y_i)$. We can conclude that two elements $(x,y),(x',y') \in Y$ are in the same orbit if and only if they have the same multiset of pairs $\{(x_1,y_1), (x_2,y_2), \ldots, (x_n, y_n) \} = \{(x'_1,y'_1), (x'_2,y'_2), \ldots, (x'_n, y'_n) \}$.

    The problem thus is reduced to counting the number of distinct multisets of size $n$ of pairs $(a,b) \in \mathbb{N}^2$, denoted by $M = \{(a_1, b_1), \ldots, (a_n, b_n)\}$, such that the component-wise sum is at most $d$. That is, $\sum_{i=1}^n a_i \leq d$ and $\sum_{i=1}^n b_i \leq d$.

    The number of such multisets is related to the number of multisets of pairs whose component sums are fixed. Let $p_2(k,\ell)$ be the \emph{number of partitions of the bi-integer $(k, \ell)$}, that is, the number of multisets of pairs $(a_j, b_j) \in \mathbb{N}^2$ such that $\sum_j a_j = k$ and $\sum_j b_j = \ell$. It follows that the total number of possible multisets whose component sums are at most $d$, without restriction on the size of the multiset, is given by
    \begin{equation}\label{eq:symmetric_group_on_pairs_upper_bound_by_partitions}
        \sum_{k=0}^d \sum_{\ell = 0}^d p_2(k,\ell).
    \end{equation}
    For $n$ large enough, any multiset counted by this sum can be augmented with $(0,0)$ pairs to form a multiset of size $n$ that satisfies the sum constraints. Therefore, the number of multisets of size $n$ is bounded by this sum. 
    We emphasize that \cref{eq:symmetric_group_on_pairs_upper_bound_by_partitions} gives the exact number of distinct orbits for $n \geq 2d$. When $n < 2d$, this equation provides an upper bound. To maintain clarity, we will concentrate on the former case ($n \geq 2d$).

    Furthermore, if $n \geq 2d$, then the upper bound in \cref{eq:symmetric_group_on_pairs_upper_bound_by_partitions} does not depend on $n$. Indeed, first observe that the number of elements needed for a partition of $(h, \ell)$ is at most $h + \ell \leq 2d$. Moreover, we can assume to have a partition $\left( (h_1,\ell_1), (h_2, \ell_2), \ldots, (h_{2d}, \ell_{2d}) \right)$ such that $\sum_{i=1}^{2d} h_i = h$ and $\sum_{i=1}^{2d} \ell_i = \ell$, where eventually we allow for the zero-pairs $(0,0)$. Finally, observe that $(h_1, \ldots, h_{2d})$ and $(\ell_1, \ldots, \ell_{2d})$ are \emph{compositions} in $2d$ nonnegative integers of $h$ and $\ell$, respectively. By a stars-and-bars argument, the number of compositions of $h$ (or $\ell$) into $2d$ nonnegative integers is $\binom{h + 2d - 1}{h}$ (or $\binom{\ell + 2d - 1}{\ell})$, thus
    \begin{equation*}
        p_2(h,\ell) \leq \binom{h + 2d - 1}{h} \binom{\ell + 2d - 1}{\ell}.
    \end{equation*}
    This upper bound depends only on $d$, and not on $n$, thus concluding our proof for the case $G = S_n$.

    Now we turn to the general case, where $G = S_{n_1} \times S_{n_2} \times \ldots \times S_{n_k}$ for some fixed positive integer $k$ and $\sum_{i=1}^k n_i = n$. The group $G$ permutes indices only within each block $\{1, \ldots, n_1\}$, $\{n_1+1, \ldots, n_1+n_2\}$, ..., $\{\sum_{i=1}^{k-1} n_i + 1, \ldots, n\}$.
    Two elements $(x,y),(x',y') \in Y$ are in the same orbit under the action of $G$ if and only if, for each block $j \in \{1, \ldots, k\}$, the multiset of pairs $\{(x_i, y_i) \mid i \text{ is in block } j\}$ is equal to the multiset of pairs $\{(x'_i, y'_i) \mid i \text{ is in block } j\}$.
    Let $M_j = \{(x_i, y_i) \mid i \text{ is in block } j\}$ be the multiset of pairs for block $j$. The size of $M_j$ is $n_j$. An orbit is uniquely determined by the $k$-tuple of multisets $(M_1, M_2, \ldots, M_k)$.

    To count the $k$-tuples of multisets $(M_1, M_2, \ldots, M_k)$ we proceed as follows. Let $0 \leq h \leq d$ and $0 \leq \ell \leq d$ be the sums of the first and second components, respectively, aggregated over all $k$ blocks. Then we consider all ways to distribute these total sums into block-specific target sums. Specifically, we consider the two ordered dispositions $(h_1,h_2 \ldots, h_k)$ and $(\ell_1, \ell_2, \ldots, \ell_k)$ of $h$ and $\ell$, respectively, in $k$ parts. Moreover, we observe that there are $\binom{h + k - 1}{k}$ and $\binom{\ell + k - 1}{k}$ dispositions, respectively. Since $d$ and $k$ are fixed, these binomial coefficients are bounded by constants depending only on $d$ and $k$.

    Next, for each block $j \in \{ 1, \ldots, k \}$, we need to determine the number of distinct multisets $M_j$ (of size $n_j$) such that the sum of its first components is exactly $h_j$ and the sum of its second components is exactly $\ell_j$. By reasoning as in the $S_n$ case, we obtain that the number of multisets whose sums are bounded by $(h_j, \ell_j)$ is bounded by $C_d = \sum_{a=0}^d \sum_{b=0}^d p_2(a,b)$, a constant depending only on $d$. Therefore, for a given set of target sums $(h_1, \ell_1), \ldots,( h_k, \ell_k)$, the number of ways to choose the $k$-tuple of multisets $( M_1, \ldots, M_k)$ is at most $(C_d)^k$.
    
    The total number of orbits is then bounded by summing over all the pairs $(h,\ell) \in [d]^2$ and then over all the possible ways to dispose $h$ and $\ell$ in $k$ parts.
    This is a constant that depends only on $d$ and $k$, and since $d$ is fixed and $k = O(1)$, this constant is also independent from $n$.
\end{proof}

\cref{th:orbits_pairs_indices_under_symmetric_action} implies that, up to symmetry, the number of distinct entries in the Gram matrix of a $G$-invariant $\sos$ polynomial is constant. This, in turn, bounds the number of variables needed to formulate the semidefinite programs in \cref{th:O(1)_symmetric_constraint_sos-proofs} and \cref{th:invariant_systems_refutation}.
On the other hand, the next \cref{th:invariant-matrix-finite-decomposition}, essentially shows that, under similar symmetry assumptions, the space of $G$-invariant matrices has constant dimension.

\begin{proposition}\label{th:invariant-matrix-finite-decomposition}
    Let $k,d \in O(1)$ be fixed positive integers. Let $G = S_{n_1} \times S_{n_2} \times \ldots \times S_{n_k}$ be a group of block permutations of indices $1, \ldots, n$.
     \begin{enumerate}
         \item If $Q\in \MS^{\omega_n^d}$ is a symmetric $G$-invariant matrix, then $Q$ can be written as a linear combination $Q = \sum_{i=1}^\ell c_i Q_i$,
        where $c_i \in \mathbb{R}$, each $Q_i$ is a symmetric matrix whose entries have value only $0$ or $1$, and $\ell$ is bounded above and independent from $n$.
        \item  If $p\in \mathbb{R}[x_1, \dots, x_n]_{d}$ is a polynomial that is invariant under the action fo $G$, then $p$ can be written as $p=\sum_{i=1}^{\ell}c_ip_i$,
         where $c_i \in \mathbb{R}$, each $p_i$ is a polynomial of degree at most $d$, whose coefficients are either $0$ or $1$, and $\ell$  is bounded above and independent from $n$.
     \end{enumerate}
\end{proposition}

\begin{proof}
    (1) Let $Q \in \MS^{\omega_n^d}$ be a symmetric $G$-invariant matrix. The invariance property $g \star Q = Q$ for all $g \in G$ is equivalent to stating that the entries of $Q$ satisfy $Q_{x^{\alpha},x^{\beta}} = Q_{g \cdot x^{\alpha}, g \cdot x^{\beta}}$ for all $g \in G$ and for all multi-indices $\alpha, \beta$ such that $\sum \alpha_i \leq d$ and $\sum \beta_i \leq d$.

    This condition implies that the value of an entry $Q_{x^{\alpha},x^{\beta}}$ is uniquely determined on the $G$-orbits of the set of pairs of multi-indices $Y = \{(\alpha, \beta) \in \mathbb{N}^n \times \mathbb{N}^n \ | \ \sum \alpha_i \leq d, \sum \beta_i \leq d \}$. That is, $Q_{x^{\alpha}, x^{\beta}} = Q_{x^{\gamma}, x^{\delta}}$ if and only if $(\alpha,\beta)$ and $(\gamma, \delta)$ belong to the same orbit of $Y$ under the action of $G$.
    
    Let $\mathcal{O}_1, \ldots, \mathcal{O}_\ell$ be the distinct orbits of $Y$ under the action of $G$. These orbits form a partition of $Y$. Let $(\gamma_i, \delta_i)$ be a chosen representative from each orbit $\mathcal{O}_i$.
    
    We define a set of matrices $Q_i \in \MS^{\omega_n^d}$ for $i=1, \ldots, \ell$. Each $Q_i$ is a 0/1 matrix for the orbit $\mathcal{O}_i$:
    \begin{align*}
        (Q_i)_{x^{\alpha},x^{\beta}} = \begin{cases}
            1 \quad \text{if } (\alpha, \beta) \in \mathcal{O}_i \\
            0 \quad \text{otherwise}
            \end{cases}
    \end{align*}
    Since the orbits $\mathcal{O}_i$ partition $Y$, for any given pair $(\alpha, \beta) \in Y$, there is exactly one orbit $\mathcal{O}_j$ such that $(\alpha, \beta) \in \mathcal{O}_j$. This means that for any $(\alpha, \beta)$, exactly one matrix $Q_j$ will have a $1$ at position $(x^\alpha, x^\beta)$, and all other $Q_i$ ($i \neq j$) will have $0$ at that position.
    
    Next, we define the coefficients $c_i$. Since $Q$ is invariant on each orbit $\mathcal{O}_i$, we can define $c_i$ as the value of $Q$ for any pair of indices in the orbit $\mathcal{O}_i$. Using the chosen representative $(\gamma_i, \delta_i)$, we set:
    \begin{align*}
        c_i = Q_{x^{\gamma_i}, x^{\delta_i}}
    \end{align*}
    Since $Q$ is constant on $\mathcal{O}_i$, for any $(\alpha, \beta) \in \mathcal{O}_i$, we have $Q_{x^{\alpha}, x^{\beta}} = c_i$.

    To show $Q = \sum_{i=1}^\ell c_i Q_i$, consider an arbitrary entry $Q_{x^{\alpha}, x^{\beta}}$. If $(\alpha, \beta) \in \mathcal{O}_j$, then $(Q_j)_{x^{\alpha}, x^{\beta}} = 1$ and $(Q_i)_{x^{\alpha}, x^{\beta}} = 0$ for $i \neq j$. Thus, $(\sum_{i=1}^\ell c_i Q_i)_{x^{\alpha}, x^{\beta}} = c_j = Q_{x^{\alpha}, x^{\beta}}$. The symmetry of each $Q_i$ follows from the symmetry of $Q$ and the orbit structure. The number of orbits $\ell$ is bounded by a constant independent from $n$, as stated in \cref{th:orbits_pairs_indices_under_symmetric_action}.

    (2) The proof uses techniques similar to those in (1). This result is commonly known as the Fundamental Theorem of Symmetric Polynomials; see \cite{Egge19} for a comparable argument and further references.
\end{proof}

\section{Automatability of $\sos$ proofs under symmetry conditions}\label{sec:symmetric-case}

In this section, we present our main results. Specifically, we establish symmetry-based conditions that ensure that property (\textsc{P}) holds. As discussed in the introduction, this implies that--under the mild assumption of Archimedeanity--finding bounded-degree $\sos$ proofs can be automated via the ellipsoid method.

We begin by outlining our approach. We consider $\mathcal{P} = \{ p_1 = 0, \ldots, p_m = 0 \}$ and a degree-$2d$ polynomial $r$. We first observe that $(r,\mathcal{P},\emptyset)$ satisfies property (P) if and only if, given that the  following system is feasible for some $Q\in \MS^n$ and $h_i\in \mathbb{R}[x_1, \dots, x_n]$ 
\begin{align}\label{proof-1}
    r=\langle Q, \x_d \x_d^\top \rangle + \sum_{i=1}^mh_ip_i , \quad \quad Q\succeq 0, \quad h_i\in \mathbb{R}[x_1, \dots,x_n]_{2d-\deg(p_i)}, 
\end{align}
then there exists a solution with $\tilde{Q} \succeq 0$ and $\tilde{h_i} \in \mathbb{R}[x_1, \ldots, x_n]$ with entries and coefficients bounded by $2^{\poly(n^d)}$. This follows from the following lemma.

\begin{lemma}\label{th:bound-coeff}
    Let $Q$ be a positive semidefinite matrix with entries bounded by $2^{\poly(n^d)}$ and let $r=\langle Q, \x_d \x_d^\top \rangle$. If there exist polynomials $s_1,\dots,s_k$ such that $r=\sum_{i=1}^k s_i^2$, then $s_i$ has coefficients bounded by $2^{\poly(n^d)}$ for every $i$.
\end{lemma}

We first recall the following. Consider a multivariate polynomial $r=\sum_{|\alpha|\leq d} c_\alpha x^\alpha$ of degree $d\in \mathbb N$. The \emph{coefficient norm} of $r$ is defined as $\|r\|_{\mathbb{R}[x]} = \max_{\alpha} \frac{|c_\alpha|}{\binom{|\alpha|}{\alpha}}$, where $\binom{|\alpha|}{\alpha} = \frac{|\alpha|!}{\alpha_1!\cdot\alpha_2!\cdot\dots \cdot\alpha_n!}$.
The coefficient norm of a polynomial can be bounded in terms of its supremum norm on $[-1, 1]^n$ as follows.
\begin{lemma}[\cite{KORDA20171}] \label{lemma-corbi}
    Let $r\in \mathbb{R}[x]_d$, then 
    \begin{equation*}
        \|r\|_{\mathbb{R}[x]} \leq 3^{d+1}\max_{x\in [-1,1]^n}|r(x)|.
    \end{equation*}
\end{lemma}
We can now prove \cref{th:bound-coeff}.

\begin{proof}[Proof of \cref{th:bound-coeff}]
    Let $s_1,\dots,s_k\in \mathbb R[x_1,\dots,x_n]$ be polynomials such that $r=\sum_{i=1}^k s_i^2$. Since the entries of $Q$ are bounded by $2^{\poly(n^d)}$, it also follows that the entries of $r=\langle Q, \x_d \x_d^\top \rangle$ are bounded by $c\, 2^{\poly(n^d)}$, where $c=O(1)$. It then follows
    \begin{align*}
        c\, 2^{\poly(n^d)} \geq \max_{x\in [-1,1]^n} |r(x)| \geq \max_{x\in [-1,1]^n} |s_i(x)^2| \quad \text{for each } i\in [k].
    \end{align*}
    Observe that, for every $i\in [k]$, it is also possible to derive $c\, 2^{\poly(n^d)} \geq \max_{x\in [-1,1]^n} |s_i(x)|$. Finally, using \cref{lemma-corbi}, we can conclude that
    \begin{align*}
        c\, 2^{\poly(n^d)} \geq \max_{x\in [-1,1]^n} |s_i(x)| \geq \|s_i\|_{\mathbb R[x]} \frac{1}{3^{d+1}} \quad \text{for each } i\in [k].
    \end{align*}
    Hence, the largest coefficient of $s_i$ is upper bounded by $d! 3^{d+1} c\, 2^{\poly(n^d)}$.
\end{proof}

We present a key technical lemma that provides bounds on SDP solutions.

\begin{lemma}\label{th:porkolab-invariant-SDP+LP} 
    Let $k_1,k_2,k_3= O(1)$ be fixed positive integers. Consider a matrix $A \in \mathbb{R}^{k_1\times (k_2+k_3)}$, symmetric matrices $Q_1, \dots , Q_{k_2}\in \MS^N$, and scalar $c\in \mathbb{R}^{k_1}$. Suppose the system
    \begin{align}\label{eq:lemma-SDP+LP-formulation}
        \sum_{i=1}^{k_2} a_iQ_i \succeq 0, \qquad
        A\Big(\begin{array}{c}
             a \\
             b
        \end{array}\Big) =c, \qquad
        a\in \mathbb{R}^{k_2},\ b\in \mathbb{R}^{k_3}
    \end{align}
    has a feasible solution. Then, it has a solution $\|(a,b) \| \leq 2^{\poly(\ell)}$, where $\ell$ is the total bit size of $Q_i, A$ and $c$.
\end{lemma}

 This lemma follows from the following classical result of Porkolab and Khachiyan \cite{PorkolabKh}, which establishes upper bounds on the magnitude of semidefinite program solutions--though such bounds are, in general, exponential in the number of variables.
\begin{theorem}[\cite{PorkolabKh}]\label{th:Porkolab_SDP_complexity}
    Any feasible system of the form
    \begin{align}\label{eq:Porkolab-thm}
        Q= Q_0 +\lambda_1Q_1+\dots +\lambda_{\ell} Q_{\ell} \succeq 0
    \end{align}
    has a solution $\lambda\in \mathbb R^{\ell}$ such that $\log \Vert \lambda \Vert = c\cdot n^{O(\min(\ell,n^2))}$, where $c \in \mathbb{N}$ is the maximum bit-length of the input coefficients and $\Vert \lambda \Vert$ is the Euclidean norm in $\mathbb{R}^{\ell}$.
\end{theorem}

\begin{proof}[Proof of \cref{th:porkolab-invariant-SDP+LP}]
    Let $a=(a_1,\dots,a_{k_2})\in \mathbb R^{k_2}$ and $b=(b_1,\dots,b_{k_3})\in \mathbb R^{k_3}$ and define $y \in \mathbb R^{k_2+k_3}$ as the vector $y=(a,b)^\top$. We encode conditions \eqref{eq:lemma-SDP+LP-formulation} in an SDP feasibility problem of the form in \ref{eq:Porkolab-thm} as follows.
    First, define the block-diagonal matrices $F_0,F_1,\dots,F_{k_2+k_3}$ of size $M=N + 2k_1$ as
    \begin{align*}
        &F_i = diag(Q_i,B_{1,i}, \dots,B_{N,i}) \in \MS^{N} \quad \text{for } i=1,\dots,k_2 \\
        &F_{k_2+j} = diag(0_{N \times N}, B_{1,k_2+j},\dots, B_{N,k_2+j}) \in \MS^{M} \quad \text{for } j=1,\dots,k_3
    \end{align*}
    for $B_{t,i}\in \mathbb R^{2\times 2}$ given by
    \begin{align*}
        B_{t,i} = \begin{pmatrix}
        0 & A_{t,i} \\
        A_{t,i} & 0
    \end{pmatrix}    \quad \text{for } t=1\dots,N \text{ and } i=1,\dots,k_2+k_3
    \end{align*}
    and 
    \begin{align*}
        F_0 = diag( 0_{N \times N},C_1,\dots,C_{k_1}), \quad \text{where, }    
        C_t=\begin{pmatrix}
            0 &  -c_t\\ 
            -c_t & 0
        \end{pmatrix} \quad \text{for } t=1\dots,k_1,   
    \end{align*}
    Then, $y$ satisfies 
    \begin{align} \label{eq:lemma-SDP+LP-formulation-proof}
        F_0 + \sum_{i=1}^{k_2+k_3}y_iF_i \succeq 0, \qquad y \in \mathbb{R}^{k_2+k_3}
    \end{align}
    if and only if both constraints \eqref{eq:lemma-SDP+LP-formulation} are satisfied. Indeed, consider a solution to \cref{eq:lemma-SDP+LP-formulation-proof}. Recall that a block diagonal matrix is PSD if and only if every diagonal block is also PSD. In particular, observe that the upper left $N \times N$-block of \cref{eq:lemma-SDP+LP-formulation-proof} corresponds exactly to $0+\sum_{i=0}^{k_2} a_i Q_i\succeq 0$, thus the PSD condition in \cref{eq:lemma-SDP+LP-formulation} holds. In addition, each $2\times 2$ remaining block is such that
    \begin{equation*}
        \begin{pmatrix}
            0 & (Ay)_j -c_j \\
            (Ay)_j -c_j  & 0
        \end{pmatrix} \succeq 0 \Longleftrightarrow (Ay)_j -c_j =0, \qquad \text{for } j=1,\dots,k_2+k_3
    \end{equation*}
    which shows the equivalence of the two formulations. The other direction, namely that a solution to \eqref{eq:lemma-SDP+LP-formulation} is also a solution to \cref{eq:lemma-SDP+LP-formulation-proof}, holds by construction.

    Since $k_2,k_3$ are fixed constants, the number $\min(k_2+k_3, n^2)$ is also a constant. By \cref{th:Porkolab_SDP_complexity}, there exists a feasible solution $\bar y$ with entries of magnitude upper bounded by $2^{\poly(\ell)}$, where $\ell$ is the total bit size of $Q_1,\dots,Q_{k_2},A,c$.
\end{proof}

We observe that the system in (\ref{proof-1}) can be viewed as a special case of the system described in (\ref{eq:lemma-SDP+LP-formulation}), where the parameters $k_1$, $k_2$, and $k_3$ are all in $O(n^d)$. This is because both the matrix $Q$ and the polynomials $h_i$ are indexed by monomials of degree at most $d$. Due to this polynomial-size dependence on $n^d$, Lemma \ref{th:porkolab-invariant-SDP+LP} does not immediately imply property (\textsc{P}). In the following, we will leverage the symmetry present in the system to address this obstacle and develop an approach that exploits this structure effectively.

\subsection{\sos\ automatability for systems of invariant polynomials}\label{sec:invariant-constraints}

We first address the case of invariant polynomials. Let $\mathcal F$ be a \GB basis and suppose there exists an $\sos$ proof of a polynomial $r\in \mathbb R[x_1,\dots,x_n]$ from the set $\mathcal P$ of equality constraints and from $\mathcal F$. \cref{th:O(1)_symmetric_constraint_sos-proofs} shows that if the polynomial $r$ and the elements of $\mathcal P$ are all invariant under the action of the direct product of symmetric groups, then there exists an $\sos$ proof of $r$ with bounded coefficients.

\begin{theorem}\label{th:O(1)_symmetric_constraint_sos-proofs}
    Let $m, d, t \in O(1)$ be fixed positive integers. Consider $\mathcal P = \{p_1 = 0, \ldots, p_m = 0\}$ set of polynomial equality constraints and let $\mathcal F=\{f_1,\dots,f_s\}$ be a Gröbner basis of $\langle \mathcal{F} \rangle$ in \texttt{grlex} order. Consider $G=S_{n_1}\times \dots \times S_{n_t}$, with $\sum_{i=1}^t n_i = n$.
    Suppose that $\mathcal F$ is a $G$-invariant system and every $p_i \in \mathcal P$ is a $G$-invariant polynomial.
    If there exists a degree-$2d$ $\sos$ proof of a $G$-invariant polynomial $r\in \mathbb R[x_1,\dots,x_n]_{2d}$ from $\mathcal P \cup \mathcal F$, then there exists a degree-$2d$ $\sos$ proof of $r$ with coefficients bounded by $2^{\poly(n^d)}$.
\end{theorem}
\begin{proof}
    By \cref{th:invariant-system-averaged-sos}, there exists a degree-$2d$ $\sos$ proof of the form
    \begin{equation}\label{eq:averaged_refutation-proof}
        r = \rho + \sum_{i=1}^m \lambda_i p_i + \sum_{i=1}^s \tilde{a_i}f_i,
    \end{equation}
    where both $\rho \in \Sigma$ and the $\lambda_i$'s are $G$-invariant.
    Next, by reducing \cref{eq:averaged_refutation-proof} modulo $\mathcal{F}$ (see \cref{th:reduction}), we obtain a solution for the following system 
     \begin{align} \label{s-1}
         \overline{r}= \overline{\rho} + \sum_{i=1}^m\overline{\lambda_ip_i},  \qquad \rho \in \Sigma_{2d}, \qquad \rho,\lambda_i \ \text{ are } G\text{-invariant for }i\in [m].  
     \end{align}
     We will show that there exists a solution to the system (\ref{s-1}) such that the coefficients are bounded by $2^{\poly(n^d)}$. This concludes the proof of the theorem, in view of \cref{th:reduction}.  

   Since $\rho$ is a $G$-invariant sum of squares, by \cref{prop:invariant-sos}, there exists a $G$-invariant matrix $P$ with $P\succeq 0$, such that $\rho = \langle P, \x_d \x_d^\top \rangle$. Then, by \cref{th:invariant-matrix-finite-decomposition}, it follows that there exists a constant $\ell_0 \in \mathbb{N}$, scalars $\gamma_{0,1} \ldots, \gamma_{0,\ell_0} \in \mathbb{R}$ and $0/1$ symmetric matrices $Q_i \in \MS^{\omega_n^d}$ such that $P = \sum_{i=1}^{\ell_0} \gamma_{0,i} Q_i$. Similarly, since the polynomials $\lambda_j$ are $G$-invariant for $j\in [m]$, by Proposition \ref{th:invariant-matrix-finite-decomposition}, there exists constants $\ell_1, \dots, \ell_m$ such that, for $j\in [m]$, $\lambda_j = \sum_{i=1}^{\ell_j} \gamma_{j,i} q_i$, where the $q_i$'s  are fixed polynomials with degree at most $\deg(\lambda_j)$ whose coefficients are $0$ or 1, and $\gamma_{j,i}$ are real scalars. 
   Therefore, for $j\in [m]$ and $i\in[\ell_j]$, we have $\lambda_jp_j= \sum_{i=1}^{\ell_j} \gamma_{j,i}(q_ip_j)$. We now define, for $j\in [m]$ and $i\in[\ell_j]$, the matrix $Q_{i,j}\in \MS^{\omega_n^d}$ such that $q_ip_j=\langle Q_{ij}, \x_d\x_d^T\rangle$. We observe that $Q_{i,j}$ can be naturally constructed so that the entries have bit size $\poly(n^d)$.
    Therefore, system (\ref{s-1}) can be rewritten as 
        \begin{align}
           & \overline{r} = \langle \sum_{i=1}^{\ell_0} \gamma_{0,i} Q_i, \overline{\x_d\x_d^\top}\rangle + \sum_{j=1}^m \sum_{i=1}^{\ell_j} \langle \gamma_{j,i} Q_{i,j}, \overline{\x_d\x_d^\top}\rangle, \label{new-recuced-1} \\
           &  \sum_{i=1}^{\ell_0} \gamma_{0,i} Q_i \succeq 0.\label{new-reduced-2}
        \end{align}
    Here, $\overline{\x_d\x_d^\top}$ denotes the matrix obtained by reducing entry-wise the matrix $\x_d\x_d^\top$ modulo $\mathcal{F}$. It remains to show that there exists a feasible solution with $|\gamma_{j,i}|< 2^{\poly(n^d)}$ (for $j=0,1,\dots, m$ and $i\in [\ell_j])$. For this, we apply \cref{th:porkolab-invariant-SDP+LP}. Observe that system \eqref{new-recuced-1}-\eqref{new-reduced-2} takes the form
       \begin{align*}
            \sum_{i=1}^{k_2} a_iQ_i \succeq 0, \quad
            A\Big(\begin{array}{c}
                 a \\
                 b
            \end{array}\Big) =c, \quad
            a\in \mathbb{R}^{k_2},\ b\in \mathbb{R}^{k_3}
        \end{align*}
    as the system in \cref{th:porkolab-invariant-SDP+LP}, where the vectors of variables $a$ and $b$ correspond to the vectors formed, respectively, by the variables $\gamma_{0,i}$ (for $i\in [\ell_0])$ and $\gamma_{j,i}$ (for $j \in [m]$, $i \in [\ell_j])$. The matrix $A$ is given by the linear equations obtained by equating the coefficients in \eqref{new-recuced-1}. Thus, $k_1, k_2, k_3$ are constant, as the number of variables and the size of the matrices are constant. The vector $c$ corresponds to the vector of coefficients of $\overline{r}$, and thus it has bit size polynomial in $n$. The bit size of the corresponding matrix $A$ is polynomial in $n$ as the bit sizes of the matrices $Q_i$ and $Q_{i,j}$ are also polynomial is $n$. Moreover, the system (\ref{new-recuced-1})-(\ref{new-reduced-2}) is feasible as the system (\ref{s-1}) is feasible. Therefore, by \cref{th:porkolab-invariant-SDP+LP}, there exists a solution with $|\gamma_{j,i}|< 2^{\poly(n^d)}$ (for $j=0,1,\dots, m$ and $i\in [\ell_j])$. Thus, by \cref{th:bound-coeff}, we obtain a feasible solution to system~\eqref{s-1} with coefficients bounded by $2^{\poly(n^d)}$. The result follows from \cref{th:reduction}.
\end{proof}

\subsection{Automatability of $\sos$ refutations for invariant polynomial systems}

In this section, we extend our analysis to a broader class of invariant systems. This generalizes the setting of \cref{sec:invariant-constraints}, which focused exclusively on invariant constraints. Rather than requiring pointwise invariance, here we allow the polynomials to be permuted among themselves under the group action. This broader scope introduces new technical challenges. Unlike \cref{th:O(1)_symmetric_constraint_sos-proofs}, which applies to $\sos$ proofs of symmetric polynomials, \cref{th:invariant_systems_refutation} applies only to refutations, and restricts to the finite domain setting, where variables range over a finite set.
Despite these limitations, we show that under symmetry assumptions, even for this more general class of unsatisfiable constraints, we can bound the coefficients appearing in $\sos$ refutations by $2^{\poly(n^d)}$. A key ingredient in our proof is reduction to normal form due to \cref{th:normal_form_archimedean_systems}.

\begin{theorem}\label{th:invariant_systems_refutation}
    Let $d,k,t,z \in O(1)$ be positive integers. Let $G = S_{n_1} \times \ldots \times S_{n_t}$, with $\sum_{i=1}^t n_i = n$. Let $\mathcal P = \{p_1 = 0, \ldots, p_m = 0\}$ be a $G$-invariant polynomial system such that $|\mathcal{P}/G|=z$ and let $\mathcal{D}$ be a finite domain constraint set of size $2k$. If there exists a degree-$2d$ $\sos$ refutation of $\mathcal{P \cup D}$, then there exists a degree-$2(d + k - 1)$ $\sos$ refutation of $\mathcal{P \cup D}$ with coefficients bounded by $2^{\poly(n^d)}$.
\end{theorem}
\begin{proof}
    Since $\mathcal{D} = \{D_i = 0\}_{i \in [n]}$ is $2k$-Archimedean (see \cref{ex:finite-domain}), by \cref{th:normal_form_archimedean_systems} there exists a degree-$2(d+k-1)$ $\sos$ refutation of the form
    \begin{align}\label{proof-last-appendix}
        -1 = \rho + \sum_{i=1}^m c_i p_i^2 + \sum_{i=1}^{n}r_iD_i
    \end{align}
    where $\rho \in \Sigma$ has degree $2(d+k-1)$, $c_i\in \mathbb{R}$, and $r_i$ are polynomials of degree  at most $2d-2$. Let $\mathcal{O}_1, \dots, \mathcal{O}_z$ be the orbits of $\mathcal{P}$ under the action of $G$.  By applying the Reynolds operator at both sides of the equality we obtain a proof of the form
     \begin{align}\label{appendix-avg-sos-ref-last-th-first}
        -1 = \sigma + \sum_{i=1}^m \tilde{c_i} p_{i}^{2} + \sum_{i=1}^{n} \tilde{r_i}D_i
    \end{align}
    where $\sigma$ is $G$-invariant. 
    We note that the number of different coefficients $\tilde{c_i}$ depends only on the number $z$ of $G$-orbits of $\mathcal{P}$. Indeed, consider the $G$-orbit $[p_i] = \{p_{i_1}, \ldots, p_{i_w}\} \in \mathcal{P}/G$ represented by polynomial $p_i \in \mathcal{P}$. We proceed to demonstrate that, in \cref{appendix-avg-sos-ref-last-th-first}, $\tilde{c}_{j_1} = \tilde{c}_{j_2}$ for all $j_1, j_2 \in \{i_1, \ldots, i_w\}$. 
    Let $G_{v,i}$ with $v \in \{i_1, \ldots, i_w\}$ be the subset of $G$ such that for every $g \in G_{v,i}$ we have that $g \cdot p_v = p_i$. Further, we observe that $\sum_{v \in \{i_1, \ldots, i_{w}\}} |G_{v,i}| = |G|$. Then, as a result of the averaging in \cref{appendix-avg-sos-ref-last-th-first}, we obtain
    \begin{align*}
        \tilde{c}_i = \frac{1}{|G|} \sum_{v \in \{i_1, \ldots,i_{w}\}} \sum_{g \in G_{v,i}} c_v = \frac{1}{|G|} \sum_{v \in \{i_1, \ldots, i_{w}\}} |G_{v,i}| c_v.
    \end{align*}
    Therefore, to show that $\tilde{c}_{j_1} = \tilde{c}_{j_2}$ for all $j_1, j_2 \in \{i_1, \ldots, i_{w}\}$ it suffices to show that $|G_{v_1,i}| = |G_{v_2,i}|$ for every $v_1, v_2 \in \{i_1, \ldots, i_{w}\}$, as this implies that $|G_{v,i}| = |G|/w$ for every $v \in \{i_1, \ldots, i_{w} \}$. This, in turn, implies that that $|G_{v,j}| = |G|/ w$ for every $v$ and $j$ in $\{ i_1, \ldots, i_w\}$, and thus that $\tilde{c}_{j_1} = \tilde{c}_{j_2}$ (note that, by the Orbit-Stabilizer Theorem applied to $p_i$, it follows that $|G|$ is divisible by $w$).
    Let $v_1,v_2 \in \{i_1, \ldots, i_{w}\}$, we will demonstrate that there exists a one-to-one correspondence between $G_{v_1, i}$ and $G_{v_2, i}$. Indeed, let $\bar{g} \in G_{v_2,i}$ be a permutation such that $\bar{g}(v_1) = v_2$ and $\bar{g}(v_2) = i$. Note that such a permutation exists since $p_i$ and $p_{v_2}$ belong to the same $G$-orbit. Let $f:G_{v_1,i} \rightarrow G_{v_2,i}$ such that $f(g) = \bar{g} \circ g$, then $f(g) \in G_{v_2,i}$. Furthermore, since $\bar{g}$ and $g$ are both bijections of $\{i_1, \ldots, i_w \}$, it follows that also $f$ is a bijective function, thus $|G_{v_1,i}| = |G_{v_2,i}|$. We can conclude that $\tilde{c}_{j_1} = \tilde{c}_{j_2}$ for all $j_1,j_2 \in \{i_1, \ldots, i_w\}$ as argued earlier. Then, we have a refutation of the form
    \begin{align}\label{appendix-avg-sos-ref-last-th}
        -1 = \sigma + \sum_{i=1}^z \tilde{c_i} \Big(\sum_{p_i\in \mathcal{O}_i} p_i^2\Big) + \sum_{i=1}^{n} \tilde{r_i}D_i
    \end{align}
    Next, observe that $\mathcal{D}$ forms a \GB basis of $\langle \mathcal{D} \rangle$ with respect to any monomial order. Then, we can reduce \cref{appendix-avg-sos-ref-last-th} modulo $\mathcal{D}$ to obtain a feasible solution of the following system
    \begin{align}\label{system-1-appendix}
         -1 = \overline{\sigma} + \sum_{i=1}^z  \tilde{c_i} \Big(\overline{\sum_{p_i\in \mathcal{O}_i}p_i^2}\Big) \quad 
         \sigma \in \Sigma_{2(d+k-1)}, \quad \tilde{c_i}\in \mathbb{R} \text{ for } i\in [z] \quad
         \sigma \text{ is } G\text{-invariant.} 
    \end{align}
    Now, we show that there exists a solution to this system such that the coefficients of $\sigma$ and $\tilde{c_i}$ (for $i\in [z]$) are bounded by $2^{\poly(n^d)}$. This will conclude the proof in view of Lemma \ref{th:reduction}.

    Since $\sigma \in \Sigma$ is $G$-invariant, by \cref{prop:invariant-sos}, there exists a $G$-invariant matrix $P$ with $P\succeq 0$, such that $\sigma = \langle P, \x_d \x_d^\top \rangle$. Then, by \cref{th:invariant-matrix-finite-decomposition}, it follows that there exists a constant $\ell \in \mathbb N$, scalars $\gamma_{1} \ldots, \gamma_{\ell} \in \mathbb{R}$ and $0/1$ symmetric matrices $Q_i \in \MS^{\omega_n^d}$ such that $P = \sum_{i=1}^{\ell} \gamma_{i} Q_i$. We let $Q_1', \dots, Q_z'$ be symmetric matrices such that $\langle Q_i', \x_d\x_d^\top \rangle=\sum_{p_i\in \mathcal{O}_i}p_i^2$. It is clear that these matrices can be picked so that their entries can be encoded with $\poly(n^d)$ bits. Therefore, system (\ref{system-1-appendix}) can be rewritten as 
    \begin{align}
       & -1 = \langle \sum_{i=1}^{\ell} \gamma_{i} Q_i, \overline{\x_d\x_d^\top}\rangle + \sum_{j=1}^z  \langle \tilde{c_i} Q_{j}', \overline{\x_d\x_d^\top}\rangle \label{new-proof1-appendix},\\
       &  \sum_{i=1}^{\ell} \gamma_{i} Q_i \succeq 0. \label{new-proof2-appendix}
    \end{align}
    Here, $\overline{\x_d\x_d^\top}$ denotes the matrix obtained by reducing entry-wise the matrix $\x_d\x_d^\top$ by the \GB basis $\mathcal{D}$. It remains to show that there exists a feasible solution with $|\gamma_{i}|< 2^{\poly(n^d)}$ (for $i \in [\ell]$) and $|\tilde{c_i}|<2^{\poly(n^d)}$ (for $i \in [z]$). For this, we apply \cref{th:porkolab-invariant-SDP+LP}. We observe that system (\ref{new-proof1-appendix})-(\ref{new-proof2-appendix}) takes the form
   \begin{align*}
        \sum_{i=1}^{k_2} a_iQ_i \succeq 0, \quad
        A\Big(\begin{array}{c}
             a \\
             b
        \end{array}\Big) =c', \quad
        a\in \mathbb{R}^{k_2},\ b\in \mathbb{R}^{k_3}
    \end{align*}
    as the system in \cref{th:porkolab-invariant-SDP+LP}, where the vectors of variables $a$ and $b$ correspond to the vectors formed, respectively, by the variables $\gamma_{i}$ (for $i\in [\ell])$ and $\tilde{c_{i}}$ (for $i\in [z]$). The matrix $A$ is given by the linear equations obtained by equating the coefficients in \eqref{new-proof1-appendix}. 
    Thus, $k_1, k_2, k_3$ are constant, as the number of variables and the sizes of the matrices are constant. The vector $c'$ corresponds to the coefficient vector of polynomial $-1$, and thus it has bit size polynomial in $n$. The bit size of the corresponding matrix $A$ is polynomial in $n$ as the bit size of the matrices $Q_i$ is also polynomial in $n$. Moreover, the system (\ref{new-proof1-appendix})-(\ref{new-proof2-appendix}) is feasible as the system (\ref{system-1-appendix}) is feasible. Therefore, by \cref{th:porkolab-invariant-SDP+LP}, there exists a solution with $|\gamma_{i}|< 2^{\poly(n^d)}$ (for $i \in [\ell]$) and $|\tilde{c_i}|<2^{\poly(n^d)}$ (for $i\in [z]$). Thus, by \cref{th:bound-coeff}, we obtain a feasible solution to system \eqref{system-1-appendix} with coefficients bounded by $2^{\poly(n^d)}$. The result follows from \cref{th:reduction}.
\end{proof}

\section{Future directions}

Our results on the existence of small-coefficient $\sos$ proofs under symmetry assumptions suggest several promising themes for further investigation. A possible direction is to extend the requirement that $m=O(1)$ in \cref{th:O(1)_symmetric_constraint_sos-proofs} to settings where the number of constraints grows with $n$, and possibly $m = poly(n)$.

Another possible direction is to exploit the normal-form reductions more broadly. While \cref{rem:normal_form_SDP_dimension_reduction} highlights potential computational benefits of the normal form in Archimedean systems,
one can apply these insights to degree-automatability questions for combinatorial instances.
For example, one could analyze how the structure from \cref{th:normal_form_archimedean_systems} (normal forms for refutations in Archimedean pairs) influences the $\sos$ bit size needed to refute instances of \textsc{3-LIN(2)}, or other constraint satisfaction problems that require super-constant $\sos$ degree for refutation. Such analysis would clarify the power and limits of the $\sos$ hierarchy and SDP relaxations by determining whether canonical representations lead to new automatability criteria or fundamental combinatorial barriers.

\bibliography{references}

\end{document}